\documentclass[12pt]{article}
\usepackage{fullpage}
\usepackage[authoryear]{natbib}
\usepackage[utf8]{inputenc}
\usepackage{amsmath}
\allowdisplaybreaks
\usepackage{amssymb}
\usepackage{amsfonts}
\usepackage{bbm}
\usepackage{graphicx}
\usepackage{float}
\usepackage{caption}
\usepackage{subcaption,mathrsfs,multirow}
\usepackage[ruled]{algorithm2e}
\usepackage{array}
\usepackage{bm}
\usepackage{authblk}
\usepackage{booktabs}
\usepackage{amsthm}

\newtheorem{theorem}{Theorem}
\newtheorem{lemma}{Lemma}

\usepackage{verbatim}
\usepackage{listings}
\usepackage{color}
\usepackage{epsfig}
\usepackage{wasysym}

\makeatletter
\newcommand*\rel@kern[1]{\kern#1\dimexpr\macc@kerna}
\newcommand*\widebar[1]{%
	\begingroup
	\def\mathaccent##1##2{%
		\rel@kern{0.8}%
		\overline{\rel@kern{-0.8}\macc@nucleus\rel@kern{0.2}}%
		\rel@kern{-0.2}%
	}%
	\macc@depth\@ne
	\let\math@bgroup\@empty \let\math@egroup\macc@set@skewchar
	\mathsurround\z@ \frozen@everymath{\mathgroup\macc@group\relax}%
	\macc@set@skewchar\relax
	\let\mathaccentV\macc@nested@a
	\macc@nested@a\relax111{#1}%
	\endgroup
}
\makeatother

\usepackage{color}

\newcommand{\comout}[1]{}

\def\lpbox#1{\vskip1mm \begin{center}
        \hspace{.0\textwidth}\vbox{\hrule\hbox{\vrule\kern6pt
\parbox{.95\textwidth}{\kern6pt \blue #1 (LP)\vskip6pt}\kern6pt\vrule}\hrule}
        \end{center} \vskip-5mm}

\usepackage[dvipsnames,svgnames,table]{xcolor}
\usepackage[colorlinks=true,
            linkcolor=red,
            urlcolor=blue,
            citecolor=blue]{hyperref}

\usepackage{setspace}

\definecolor{lbcolor}{rgb}{0.95,0.95,0.95}
\definecolor{darkred}{RGB}{150,50,50}
\definecolor{brown}{RGB}{250,100,100}
\definecolor{green}{RGB}{000,150,100}
\definecolor{purple}{RGB}{200,000,250}

\def\blue{\color{blue}}

\title{\large \bf Deep learning for interval-censored failure time data from case-cohort studies}
\author[1]{Yeyu Xiao}
\author[1,*]{Yonghong Long}
\affil[1]{School of Mathematics, Renmin University of China, Beijing, 100872, China}

\date{} 

\begin{document}
\maketitle
\vspace{-15pt}

\makeatletter
\let\thefootnote\relax\footnotetext{*Corresponding author: longyh@ruc.edu.cn}
\makeatother

\vspace{-5pt}
\begin{abstract}
\setstretch{1.5}

Interval-censored data are common in fields such as epidemiology and demography. When the failure event of interest is relatively rare and the collection of covariates is costly, researchers often adopt the case-cohort design to reduce study costs. However, existing studies typically rely on the assumption of linearity in modeling covariates, which may not capture the complex and nonlinear relationships present in real data. To address this limitation, we consider a class of transformation models with unspecified covariate-dependent functions. We propose a sieve maximum weighted likelihood approach for interval-censored data arising from the case-cohort design, which combines deep neural networks with Bernstein polynomials. The method employs a deep neural network to flexibly represent the covariate-dependent function and uses Bernstein polynomials to approximate the cumulative baseline hazard function. We establish the consistency and convergence rate of the proposed estimator and show that the resulting nonparametric deep neural network estimator attains the minimax optimal rate of convergence (up to a polylogarithmic factor). Simulation studies suggest that the proposed method performs well in practice. Finally, we apply the method to a real dataset and use the SHAP (Shapley Additive Explanations) approach to attribute the neural network predictions of the covariate-dependent function to covariates. The results indicate that our method is both accurate and interpretable.
\end{abstract}

\newpage
\clearpage

\setstretch{1.5}

\section{Introduction}

Interval-censored failure time data commonly arise in epidemiological, demographic, and biomedical studies. Such data occur when the failure event of interest cannot be observed exactly, but is only known to fall within a certain time interval. Right-censored data, which have been extensively studied, represent a special case of interval-censored data. However, the mechanisms that generate interval censoring are more complex, and the two types differ fundamentally. As a result, many theories and methods developed for right-censored data cannot be directly extended to interval-censored data, making the analysis of the latter more challenging \citep{Sun2006statistical}.

Some failure events are relatively rare, and a large sample size is typically required to yield reliable information about the effects of covariates on such event times \citep{Zeng2014Efficient}. In epidemiological cohort studies, the measurement of some covariates can be difficult or costly, making it impractical to collect them for all subjects. To achieve the same objectives as a cohort study under limited resources, \cite{Prentice1986case} proposed the case-cohort design, in which a random sample (subcohort) is first drawn from the entire cohort, and covariates that are expensive or difficult to obtain are collected only for subjects in the subcohort and those who experience the event of interest. Subsequently, many scholars have conducted in-depth investigations on the case-cohort design, but most of them have focused on right-censored data. In recent years, the case-cohort design based on interval-censored data have also been extensively studied. \cite{Li2011Relative} studied the case-cohort design with current status data. \cite{Zhou2017Case} fitted the proportional hazards model to general interval-censored data from the case-cohort design and developed a sieve weighted likelihood approach using inverse probability weighting. \cite{Du2021Reg} investigated the case-cohort design under informatively interval-censored data, where the censoring mechanism is not independent of the failure time. \cite{Zhou2021Semi} and \cite{Lou2023Semi} considered the case-cohort design for multivariate interval-censored data and generalized it to non-rare events. However, these studies typically restrict the covariate-dependent function to a linear form. In practice, imposing such parametric assumptions may be overly restrictive, as covariate-dependent functions often exhibit complex nonlinear patterns. To capture these complex relationships, a tool capable of flexibly approximating intricate nonlinear functions is required, and deep neural networks offer a powerful solution for this purpose.

Neural networks are functions composed of multiple layers, and each layer performs a linear transformation followed by a nonlinear activation function (e.g., the ReLU function). Shallow neural networks have been shown to approximate any continuous function with arbitrary accuracy \citep{Cybenko1989Approximation, Leshno1993Multilayer}, while deeper architectures can achieve similar performance with fewer parameters \citep{Telgarsky2016benefits}. Moreover, deep neural networks are capable of representing a rich class of functions and can identify low-dimensional structures in high-dimensional data, thus alleviating the curse of dimensionality. The many favorable properties of neural networks have motivated researchers to explore the application of deep learning methods to survival analysis. \cite{Faraggi1995neural} were the first to replace the linear predictor in the Cox proportional hazards model with a nonlinear function output by a shallow neural network. \cite{Katzman2018DeepSurv} developed a Cox proportional hazards deep neural network, called DeepSurv. However, these approaches are limited to right-censored data. For interval-censored data, \cite{Sun2023Neural} proposed a neural network method that innovatively incorporates Bernstein polynomials within the network framework to estimate the cumulative baseline hazard function. Although these methods have achieved considerable progress in applications, their theoretical understanding remains limited. Motivated by recent theoretical advances in deep learning for nonparametric regression, \cite{Zhong2021NEURIPS,Zhong2022Deep} provided theoretical support for deep learning methods applied to right-censored data. In 2021, they considered the deep extended hazard model and derived the consistency and convergence rate of the survival function estimator. In 2022, they studied the deep partially linear Cox model (DPLCM) and developed optimal asymptotic theory for both the parametric and nonparametric components. \cite{Wu2024Deep} considered the DPLCM under current status data and proved the corresponding asymptotic properties. \cite{DU2024Deep} further extended the DPLCM to case II interval-censored data.

Although neural networks exhibit excellent predictive performance, their complex structure and large number of parameters have led them to be regarded as “black-box” models with limited interpretability. Interpreting model predictions is essential for enhancing the credibility of such models. \cite{Zhong2022Deep} categorized covariates into two groups; they modeled the effects of the treatment covariates with a linear function and learned the complex relationships of other covariates through a neural network. However, the study did not analyze the contributions of the covariates modeled by the neural network to the predictions, which limits the interpretability of the model. \cite{Sun2020Genome} applied the LIME (local interpretable model-agnostic explanation) method to interpret the predictions of their DNN-based survival model on test set samples. SHAP is a model interpretation method based on the game-theoretic Shapley value. Within the feature attribution framework, it provides rigorous theoretical guarantees compared to LIME \citep{Lundberg2017unified}. The SHAP method supports local explanations of model predictions and can derive global variable importance by aggregating local contributions. For neural network models, it also allows indirect inference of variable interaction effects. Therefore, applying SHAP facilitates the interpretation of our model’s predictions and helps uncover new patterns from the data.

In this paper, we consider the case-cohort design for interval-censored data as well as the generalized case-cohort design for non-rare events. For the data arising from these designs, we focus on a class of transformation models that encompasses many commonly used models, such as the proportional hazards model and the proportional odds model, thus providing greater flexibility. Restricting the effects of some covariates to a linear form may obscure potential nonlinear relationships and interactions in the data. Therefore, we model the effects of all covariates uniformly as an unknown smooth function. The resulting model comprises two nonparametric components: the time-dependent, infinite-dimensional cumulative baseline hazard function, approximated using a sieve method with Bernstein polynomials; and the covariate-dependent function, modeled flexibly via a neural network. In addition, we construct the likelihood function employing inverse probability weighting (IPW), which can account for the sampling bias induced by the (generalized) case-cohort design. In the theoretical analysis, we establish the consistency and convergence rate of the proposed estimator and show that our nonparametric deep neural network estimator obtains the minimax optimal rate of convergence (up to a polylogarithmic factor). Moreover, we demonstrate the performance of the proposed method through extensive simulation studies. In real data analysis, we employ the SHAP method to attribute the predictions of the neural network that models the covariate-dependent function. This allows us to quantify the global contribution of each covariate and the distribution of their effects on individual predictions, as well as to explore potential interactions among covariates through dependence plots. 

The remainder of this paper is organized as follows. Section~\ref{mde} describes the model, designs and data structure, and introduces the proposed estimation procedure. Section~\ref{theor} discusses the theoretical properties of the resulting estimator. In Section~\ref{simul}, we conduct extensive simulation studies to evaluate the performance of the proposed method. In Section~\ref{applic}, we apply the proposed approach to a children mortality study in Nigeria; the results demonstrate that our method is accurate and interpretable. Section~\ref{concl} provides concluding remarks and discussion. All technical proofs are given in the Appendix.

\section{Model, Data and Estimation}
\label{mde}

The (generalized) case–cohort design is essentially a two-phase sampling design. Consider a cohort study with $n$ independent subjects. For the $i$th subject, the failure time is denoted by $T_{i}$, and $Z_{i}$ represents the associated $p$-dimensional covariate vector. Let $U_{i1}, \ldots, U_{iK_{i}}$ denote the random examination times for subject $i$, satisfying $0 = U_{i0} < U_{i1} < \cdots < U_{iK_{i}} < U_{iK_{i}+1} = \infty$, where the number of examinations $K_{i}$ is a positive integer. We assume that, conditional on the covariates, the examination times are independent of the failure time. Since continuous monitoring is not available, we only know that the failure event for subject $i$ occurs within an observation interval $(L_{i}, R_{i}]$, with $L_{i}= \max\left\{ U_{ik}: U_{ik}< T_{i}, k=0, \ldots, K_{i}\right\}$, $R_{i}= \min\left\{ U_{ik}: U_{ik}\geq T_{i}, k=1, \ldots, K_{i}+1\right\}$. We then define $\Delta_{iL} = I(L_{i} = 0)$ and $\Delta_{iI} = I(L_{i} \neq 0, R_{i} < \infty)$, where $I(\cdot)$ denotes the indicator function. When $\Delta_{iL} = 1$, it indicates that the failure event of subject $i$ occurs before the first examination, so the failure time is left-censored. When $\Delta_{iL} + \Delta_{iI} = 0$, it indicates that the failure event of subject $i$ occurs after the last examination, so the failure time is right-censored. Thus, at Phase I, we observe the interval-censored data for all $n$ subjects,
\begin{equation*}
\left\{L_{i},R_{i},\Delta_{iL},\Delta_{iI} \right\}, i = 1, \ldots , n.
\end{equation*} 

At Phase II, we first draw a subcohort from the study cohort by independent Bernoulli sampling with known probability $p_{s}\in (0,1]$. Let $\zeta_{i} = 1$ indicate that subject $i$ is included in the subcohort, and $\zeta_{i} = 0$ otherwise. Subsequently, we draw a subset of cases from those cases not included in the subcohort (i.e., subjects with $\Delta_{iL} + \Delta_{iI} = 1$ and $\zeta_{i} = 0$) through independent Bernoulli sampling with known probability $p_{c} \in (0,1]$. Let $\xi_{i} = 1$ indicate that subject $i$ is selected into the case subset, and $\xi_{i} = 0$ otherwise. Finally, expensive covariate measurements are performed only for subjects in the subcohort (i.e., $\zeta_{i} = 1$) and those in the case subset (i.e., $\xi_{i} = 1$). For rare failure events, we set $p_{c} = 1$, which means that all cases are selected; this corresponds to the case-cohort design. For non-rare or not-so-rare failure events, we set $p_{c} \in (0,1)$, so that only a subset of cases not included in the subcohort is selected; this corresponds to the generalized case-cohort design. Under the (generalized) case-cohort design, the observed data can be represented as:
\begin{equation*}
O_{i}^{\varphi  } =\left\{ L_{i},R_{i},\Delta_{iL},\Delta_{iI},\varphi  _{i}Z_{i},\varphi  _{i}\right\}, i = 1, \ldots , n.
\end{equation*} 
Here, $\varphi_{i} = 1$ indicates that the covariates of subject $i$ are obtained, and $\varphi_{i} = 0$ otherwise.

Assume that the failure time $T$ follows the transformation model with an unspecified covariate-dependent function. The conditional cumulative hazard function of $T$ given the covariate vector $Z \in \mathbb{R}^{p}$ takes the form:
\begin{equation}
\label{equa1}
\Lambda (t|Z)=G(\Lambda(t)\exp(g(Z))),
\end{equation} 
where $\Lambda$ is an unspecified cumulative baseline hazard function, $g: \mathbb{R}^{p} \to \mathbb{R}$ is an unknown function, and $G$ is a prespecified strictly increasing function. Two-phase sampling induces sampling bias, which can be addressed using inverse probability weighting. To estimate $\theta = (\Lambda, g)$, the inverse probability weighted log-likelihood function takes the form:
\begin{equation}
\label{equa2}
\begin{split}
l_{n}^{w}\left ( \Lambda ,g \right ) &= \sum_{i=1}^n w_i \big\{ \Delta_{iL} \log \left[ 1 - \exp\left(-G\left(\Lambda(R_i) e^{g(Z_i)}\right)\right) \right] \\
&\quad + \Delta_{iI} \log \left[ \exp\left(-G\left(\Lambda(L_i) e^{g(Z_i)}\right)\right) - \exp\left(-G\left(\Lambda(R_i) e^{g(Z_i)}\right)\right) \right] \\
&\quad - (1 - \Delta_{iL} - \Delta_{iI}) G\left(\Lambda(L_i) e^{g(Z_i)}\right) \big\}.
\end{split}
\end{equation}
As in \cite{Zhou2024Improving}, the weight $w_i$ is set as
\begin{equation*}
w_i = \frac{\varphi _i}{\pi_p(\Delta_{iL},\Delta_{iI})} = \frac{\varphi _i}{\left(1-\Delta_{iL}-\Delta_{iI}\right)p_s + \left(\Delta_{iL}+\Delta_{iI}\right)(p_s+(1-p_s)p_c)}.
\end{equation*}

We now turn to the estimation of the unknown functions $g$ and $\Lambda$ in model \eqref{equa1}. We approximate the covariate-dependent function $g$ using a neural network, and briefly introduce the relevant concepts of deep neural networks (DNNs) as function approximation tools. An $(H+1)$-layer DNN with layer width $\boldsymbol{p}$ is a composite function $g : \mathbb{R}^{p_0} \to \mathbb{R}^{p_{H+1}}$, defined recursively as follows:
\begin{equation}
\label{equa3}
\begin{aligned}
&g(z) = W_Hg_H(z) + v_H, \\
&g_H(z) = \sigma(W_{H-1}g_{H-1}(z) + v_{H-1}), \ldots ,g_1(z) = \sigma(W_0 z + v_0),
\end{aligned}
\end{equation}
where $H\in \mathbb{N}_{+}$ denotes the number of hidden layers and $\boldsymbol{p} = \left ( p_0,...,p_H,p_{H+1} \right )\in \mathbb{N}_{+}^{H+2}$ specifies the width of each layer (i.e., the number of neurons). The matrices $W_h\in \mathbb{R}^{p_{h+1}\times p_h}$ and vectors $v_h\in \mathbb{R}^{p_{h+1}}$ (for $h = 0,\ldots,H$) are the parameters of the DNN, where $(W_h)_{i,j}$ represents the weight connecting the $j$th neuron in layer $h$ to the $i$th neuron in layer $h+1$, and $(v_h)_i$ denotes the bias term associated with the $i$th neuron in layer $h+1$. The activation function $\sigma$ is chosen a priori and is applied componentwise to vectors, that is, $\sigma((z_1,\ldots,z_{p_h})^\top) = (\sigma(z_1),\ldots,\sigma(z_{p_h}))^\top$. In this paper, we employ the ReLU activation function: $\sigma(z) = \max\left\{ z, 0\right\}$.

Given $H \in \mathbb{N}_{+}$ and $\boldsymbol{p}\in \mathbb{N}_{+}^{H+2}$, a class of DNN can be expressed as:
\begin{equation*}
\begin{split}
\mathcal{G}(H, \boldsymbol{p}) = \big\{ g : & \text{$g$ is a DNN with $(H + 1)$ layers and width vector $\boldsymbol{p}$ such that } \\
& \max \left\{ \|W_h\|_\infty, \|v_h\|_\infty \right\} \leq  1, \text{ for all } h = 0, \dots, H \big\},
\end{split}
\end{equation*}
where $\left\| \cdot \right\|_{\infty }$ denotes the supremum norm of a matrix or vector. It is well known that deep feedforward networks with fully connected layers often involve a large number of parameters, which may lead to overfitting. Pruning weights can reduce the total number of nonzero parameters and lead to sparse connections across layers. This strategy can mitigate overfitting to some extent \citep{Han2015Learning, Schmidt-Hieber2020deep}. Based on this idea, for $s \in \mathbb{N}_{+}$ and $D > 0$, a class of sparse neural networks can be represented as:
\begin{equation}
\label{equa4}
\mathcal{G}(H, s, \boldsymbol{p}, D) := \left\{ g \in \mathcal{G}(H, \boldsymbol{p}) : \sum_{h=1}^H  \|W_h\|_0 + \|v_h\|_0  \leq s, \|g\|_\infty \leq  D \right\},
\end{equation}
where $\|g\|_\infty$ denotes the supremum norm of the function $g$, and $\|\cdot\|_0$ represents the number of nonzero elements in a matrix or vector.

Next, we consider estimating the unspecified cumulative baseline hazard function $\Lambda \in \mathcal{M}_0$, where $\mathcal{M}_0$ denotes the collection of all bounded and continuous nondecreasing, nonnegative functions over the interval $\left [ c, u \right ]$, with $0 \leq c < u < \infty $. Following \cite{Zhou2017Sieve}, we handle the infinite-dimensional parameter $\Lambda$ using a sieve approach based on Bernstein polynomials. The space $\mathcal{M}_n$ is defined as follows:
\begin{equation}
\label{equa5}
\mathcal{M}_n = \left\{ \Lambda_n(t) = \sum_{k=0}^m \phi_k B_k(t, m, c, u) :  \sum_{k=0}^m |\phi_k| \leq M_n , 0 \leq \phi_0 \leq \cdots \leq \phi_m \right\},
\end{equation}
where $B_k(t, m, c, u)$ represents the Bernstein basis polynomial, which is given by:
\begin{equation*}
B_k(t, m, c, u) = \binom{m}{k} \left( \frac{t - c}{u - c} \right)^k \left( 1 - \frac{t - c}{u - c} \right)^{m-k}, k = 0, \ldots, m,
\end{equation*}
with degree $m = o(n^{\nu})$ for some $\nu \in (0, 1)$.

For simplicity, let $\mathcal{G}=\mathcal{G}(H, s, \boldsymbol{p}, \infty)$. We obtain the estimator $\hat{\theta} = (\hat{\Lambda}_n,\hat{g})$ of $\theta =\left ( \Lambda ,g \right )$ by maximizing the weighted log-likelihood function \eqref{equa2} over the space $\mathcal{M}_n \times \mathcal{G}$
\begin{equation}
\label{equa6}
\hat{\theta} = (\hat{\Lambda}_n,\hat{g}) = \mathop{\operatorname{arg\,max}}\limits_{(\Lambda_n, g) \in \mathcal{M}_n \times \mathcal{G}} l_{n}^{w}\left ( \Lambda ,g \right ).
\end{equation}

\section{Theoretical Analysis}
\label{theor}

We now investigate the asymptotic properties of the proposed estimator $\hat{\theta}$. Some restrictions on the nonparametric function $g$ are first required. A Hölder class of smooth functions with parameters $\alpha$, $B > 0$, and domain $\mathbb{D} \subset \mathbb{R}^r$ is defined as follows:
\begin{equation*}
\mathcal{H}^{\alpha}_{r}(\mathbb{D}, B) = \left\{ g : \mathbb{D} \to \mathbb{R} : \sum_{\boldsymbol{\beta} : |\boldsymbol{\beta}| < \alpha} \|\partial^{\boldsymbol{\beta}} g\|_{\infty} + \sum_{\boldsymbol{\beta} : |\boldsymbol{\beta}| = \lfloor \alpha \rfloor} \sup_{\boldsymbol{x},\boldsymbol{y} \in \mathbb{D}, \boldsymbol{x} \neq \boldsymbol{y}} \frac{|\partial^{\boldsymbol{\beta}} g(\boldsymbol{x}) - \partial^{\boldsymbol{\beta}} g(\boldsymbol{y})|}{\|\boldsymbol{x} - \boldsymbol{y}\|_{\infty}^{\alpha - \lfloor \alpha \rfloor}} \leq B \right\},
\end{equation*}
where $\lfloor \alpha \rfloor$ is the largest integer strictly smaller than $\alpha$, $\partial^{\boldsymbol{\beta}} := \partial^{\beta_1} \ldots \partial^{\beta_r}$ with $\boldsymbol{\beta} = (\beta_1, \ldots, \beta_r)$ and $|\boldsymbol{\beta}| = \sum_{k=1}^{r} \beta_k$. Let $L \in \mathbb{N}$, $\boldsymbol{d} = (d_0, \ldots, d_{L+1}) \in \mathbb{N}^{L+2}_+$, we consider a composite Hölder function:
\begin{equation*}
g = g_L \circ g_{L-1} \circ \cdots \circ g_1 \circ g_0,
\end{equation*}
where $g_l : [a_l, b_l]^{d_l} \to [a_{l+1}, b_{l+1}]^{d_{l+1}}$, $l = 0, \ldots, L$. Denote $g_l = (g_{l1}, \ldots, g_{ld_{l+1}})^\top$ and let $d_{lj}$ be the unique number of features that each $g_{lj}$ depends on. Define $\tilde{d}_l = d_{l1} \lor \cdots \lor d_{ld_{l+1}}$. We further assume that $g$ belongs to a composite smoothness function class:
\begin{equation*}
\begin{aligned}
\mathcal{H}(L, \boldsymbol{\alpha}, \boldsymbol{d}, \boldsymbol{\tilde{d}}, B) = \biggl\{ &g = g_L \circ \cdots \circ g_0 : g_l = (g_{l}, \ldots, g_{ld_{l+1}})^\top \text{and } \\
&g_{lj} \in \mathcal{H}_{\tilde{d}_l}^{\alpha_l}([a_l, b_l]^{\tilde{d}_l}, B), \text{ for some } |a_l|, |b_l| \leq B \biggr\},
\end{aligned}
\end{equation*}
where $\boldsymbol{\alpha} = (\alpha_0, \ldots, \alpha_L) \in \mathbb{R}^{L+1}_+$ and $\boldsymbol{d} = (d_0, \ldots, d_{L+1}) \in \mathbb{N}^{L+2}_+$, $\boldsymbol{\tilde{d}} = (\tilde{d}_0, \ldots, \tilde{d}_L) \in \mathbb{N}^{L+1}_+$ with $\tilde{d_l} \leq d_l$, $l = 0, \ldots, L$. The functions in this class are characterized by two dimensions, $\boldsymbol{d}$ and $\boldsymbol{\tilde{d}}$, with $\boldsymbol{\tilde{d}}$ indicating the intrinsic dimension of the function. Furthermore, we denote $\tilde{\alpha}_l = \alpha_l \prod_{k=l+1}^L (\alpha_k \land 1)$ and $\gamma_n = \max_{l=0,\ldots,L} n^{-\tilde{\alpha}_l/(2\tilde{\alpha}_l + \tilde{d}_l)}$, where $a \land b := \min\{a, b\}$.

For any $\theta_1 = (\Lambda_1, g_1)$ and $\theta_2 = (\Lambda_2, g_2)$, define  
\begin{equation*}
d(\theta_1, \theta_2) = \left\{\| g_1 - g_2 \|_{L^2}^2 + \| \Lambda_1 - \Lambda_2 \|_{2)}^2 \right\}^{1/2},
\end{equation*}
where $\| g_1 - g_2 \|_{L^2}^2 = \mathbb{E}\left\{ g_1(Z) - g_2(Z)\right\}^2$ and $\|\Lambda_1 - \Lambda_2\|_{2}^2 = \mathbb{E}\{\Lambda_1(L) - \Lambda_2(L)\}^2 + \mathbb{E}\{\Lambda_1(R) - \Lambda_2(R)\}^2$. Let $\theta_0 = (\Lambda_0, g_0)$ denote the true value of $\theta$. 

\begin{theorem}
\label{th1}
Assume that Conditions (C1)-(C5) given in the Appendix hold. Then there exists an estimator $\hat{\theta}$ in \eqref{equa6} satisfying $\mathbb{E}\{\hat{g}(Z)\} = 0$, and this estimator converges to $\theta_0$ in probability.
\end{theorem}

\begin{theorem}
\label{th2}
Assume that Conditions (C1)-(C6) given in the Appendix hold. Then we have that
\begin{equation*}
d(\hat{\theta},\theta_0 )=O_p(n^{-r\nu/2} + \gamma_n \log^2 n)
\end{equation*}
and
\begin{equation*}
\| \hat{g} - g_0 \|_{L^2([0, 1]^p)}=O_p(\gamma_n \log^2 n),
\end{equation*}
where $\nu \in (0, 1)$ such that $m = o(n^{\nu})$ and $r$ is defined in Condition (C5).
\end{theorem}

\begin{theorem}
\label{th3}
Under Conditions (C2)-(C5) in the Appendix, there exists a constant $0 < C < \infty $, such that
\begin{equation*}
\inf_{\hat{g}} \sup_{(\Lambda_0, g_0)\in\mathcal{M}_0\times\mathcal{H}_0} \mathbb{E}\{\hat{g}(Z) - g_0(Z)\}^2 \geq C\gamma _{n}^{2} ,
\end{equation*}
where the infimum is taken over all possible estimators $\hat{g}$ based on the observed data.
\end{theorem}

Theorem~\ref{th1} establishes the asymptotic consistency of the proposed estimator $\hat{\theta}$. Theorem~\ref{th2} provides the convergence rate for $\hat{\theta}$, including its DNN-based component $\hat{g}$. Theorem~\ref{th3} further establishes the minimax lower bound for estimating $g_0$, indicating that the DNN-based estimator is rate optimal (up to a polylogarithm factor). We present the proofs of these theorems and their required regularity conditions in the Appendix.

\section{A Simulation Study}
\label{simul}

In this section, we present several results from simulation studies to evaluate the finite-sample performance of the proposed method. We assumed the covariate vector $Z=(Z_1,\ldots,Z_5)^\top$, where $Z_1$ followed the Bernoulli distribution with success probability 0.5, and $(Z_2, Z_3, Z_4)^\top$ followed the trivariate normal distribution with mean zero, variance one, and correlation $\text{Corr}(Z_i,Z_j)=0.5^{|i-j|}$ for $i,j \in {2,3,4}$. Subsequently, each component of $(Z_2, Z_3, Z_4)^\top$ was truncated to the interval $[0,2]$. In addition, $Z_5$ followed the uniform distribution on the interval $[0,1]$. The failure time $T$ was assumed to follow the transformation model \eqref{equa1}, with the cumulative baseline hazard function $\Lambda(t)=0.1t^{2}$. For the transformation function $G$, we considered the class of logarithmic transformations defined as $G(x)=\log(1+rx)/r$ for $r \geq 0$. In our simulations, we examined three cases with $r=0, 0.5,$ and $1$. Setting $r=0$ yields $G(x)=x$, under which the transformation model reduces to the proportional hazards model. When $r=1$, we have $G(x)=\log(1 + x)$, which corresponds to the proportional odds model. For the covariate-dependent function $g(Z)$, we considered the following three settings:
\begin{align*}
\text{Case 1 (Linear):} \quad & g(Z) = z_1 - 0.3 z_2 - 0.3 z_3 + 0.6 z_4 - 0.5 z_5 - 0.25, \\
\text{Case 2 (Deep 1):} \quad & g(Z) = \frac{z_1 z_2^2}{3} + \log(z_3 + 1) + \sqrt{z_3 z_4} + \frac{\exp(z_5)}{3} - 1.18, \\
\text{Case 3 (Deep 2):} \quad & g(Z) = \frac{\left( \frac{z_1 z_2^2}{3} + \log(z_3 + 1) + \sqrt{z_3 z_4} + \frac{\exp(z_5)}{3} \right)^2}{4} - 0.53.
\end{align*}
Various intercept terms, $0.25$, $1.18$, and $0.53$, were added to $g$ such that the condition $\mathbb{E}\left\{ g(Z)\right\} = 0$ holds for all covariate settings.

Interval-censored data were generated by mimicking real follow-up studies. Specifically, it was assumed that all $n$ subjects were scheduled to receive $k$ equally spaced visits within the time interval $[0,\tau]$, at times $\tau_1,\ldots,\tau_k$, with spacing $td=\tau/(k+1)$. In practice, each subject might advance, delay, or miss certain visits. Thus, for the $i$th subject, the actual visit times were given by $ \left\{ (\tau_{j} + \varepsilon _{ij})\psi_{ij},j=1,\ldots,k\right\}$ for $i=1,\ldots,n$, where $\varepsilon_{ij}$ were i.i.d. uniform random variables on $\left [ -td/3,td/3\right ]$, representing deviations in visit timing, and $\psi_{ij}$ were i.i.d. $\text{Bernoulli}(0.8)$ random variables, indicating that each scheduled visit was attended with probability $0.8$. If subject $i$ was found to have experienced the failure event at the first visit, then $L_i=0$, $R_i$ was set to the first visit time, and $(\Delta_{iL}, \Delta_{iI})=(1,0)$. If subject $i$ was observed to experience the event at any subsequent visit, $L_i$ was defined as the previous visit time, $R_i$ as the current visit time, and $(\Delta_{iL}, \Delta_{iI})=(0,1)$. If subject $i$ had not experienced the event by the last visit, $L_i$ was taken as the time of the last visit, $R_i=\infty$, and $(\Delta_{iL}, \Delta_{iI})=(0,0)$. The number of scheduled visits was fixed at $k=10$. By varying the study end time $\tau$, different event rates $p_e$ were obtained, with $p_e=0.1$, $0.2$, and $0.3$ considered.

We considered datasets with sample sizes of $2000$ and $3000$. For each dataset, a stratified split of $9{:}1$ was used to divide the data into training and testing sets, preserving the original distribution of cases. The training set served as the study cohort. Case-cohort studies for event rates $p_e=0.1$ or $0.2$ and generalized case-cohort studies for $p_e=0.2$ or $0.3$ were considered. In both designs, the subcohort was selected via Bernoulli sampling with probability $p_s=0.2$. In the generalized case-cohort design, the subset of cases outside the subcohort was selected using Bernoulli sampling with probability $p_c=0.5$.

We implemented the estimator in equation \eqref{equa6} using PyTorch \citep{paszke2019pytorch}. Following \cite{Sun2023Neural}, a custom BPNet was constructed to implement Bernstein polynomials within a neural network framework. This was combined with a deep neural network used to approximate the covariate-dependent function, and both networks were trained simultaneously. The loss function for the full model was defined as the negative weighted log-likelihood function, and parameters were updated using the Adam optimizer \citep{kingma2014adam}.

Hyperparameters are model parameters that must be specified prior to training. In the simulation study, they included the batch size, the number of hidden layers $H$ and the number of neurons $p_h$ in each hidden layer, the dropout rate \citep{srivastava2014dropout}, the learning rate \citep{Goodfellow2016Deep} for BPNet, and the learning rate for the covariate network. For simplicity, the number of neurons was assumed to be the same in each hidden layer (i.e., $p_i = p_j$ for $1 \le i,j \le H$). Before each simulation run, an additional dataset of equivalent size was generated for hyperparameter tuning. A grid search combined with ten-fold cross-validation was performed, using the negative log-likelihood as the evaluation metric. Early stopping \citep{Goodfellow2016Deep} was applied to monitor network training and the hyperparameter combination that achieved the best average performance across the ten folds, along with the corresponding average number of training epochs, was selected for the final simulation. The parameters of the DNN used to estimate the covariate function were initialized using PyTorch's default random initialization. In all simulations, the interval $[c,u]$ in the Bernstein polynomial was set to $[0,\tau]$, with degree $m=5$.

The proposed method was applied to fit the model \eqref{equa1} using the (generalized) case-cohort sample, referred to as PRO. In addition, a sieve likelihood approach that combined DNN with Bernstein polynomials was used to fit the model \eqref{equa1} based on the subcohort sample and on a simple random sample of the same size as the (generalized) case-cohort sample, denoted by SUB and SRS, respectively. Furthermore, a linear transformation model was fitted to the (generalized) case-cohort sample, called LTM. The performance of these methods in estimating the covariate-dependent function $g$ was evaluated using the relative error (RE), defined as:
\begin{equation*}
\text{RE}(\hat{g}) = \left\{ \frac{\frac{1}{n_1} \sum_{i=1}^{n_1} \{ [\hat{g}(Z_i) - \bar{\hat{g}}] - g(Z_i) \}^2 }{\frac{1}{n_1} \sum_{i=1}^{n_1} [g(Z_i)]^2} \right\}^{1/2},
\end{equation*}
where $\hat{g}$ and $g$ are evaluated on the covariates of the test set $\left\{ Z_i:i = 1,\ldots,n_1\right\}$, $n_1$ denotes the sample size of the test set, and $\bar{\hat{g}} = \sum_{i=1}^{n_1} \hat{g}(Z_i)/n_1$. The mean of $\hat{g}$ on the test set was subtracted because the solution of maximizing the weighted log-likelihood is only unique up to a constant. The mean squared prediction error (MSPE) was further used to compare the predictive accuracy of the survival function across different methods, defined as:
\begin{equation*}
L(\widehat{S}) = \frac{1}{n_1} \sum_{i=1}^{n_1} \frac{1}{a-b} \int_b^{a} \{ S(t|Z_i) - \widehat{S}(t|Z_i) \}^2 \, dt,
\end{equation*}
where $a$ and $b$ denote the maximum finite value and the minimum value of all observed $\{L_i, R_i\}$, respectively, and $n_1$ represents the sample size of the test set. The MSPE essentially measures the average $L_2$ distance between the estimated and the true survival functions in the test set. Smaller values of both metrics indicate better model performance. All results were obtained from $1000$ repetitions.

\begin{table}[htbp]
\centering
\caption{\label{table1}Relative errors (standard deviations) of $\hat{g}$ for different methods}
\resizebox{6.6in}{!}{
\begin{tabular}{ccccccccccccccc}
\toprule
\multirow{2}{*}{} & \multirow{2}{*}{$p_{\text{e}}$} & \multirow{2}{*}{$n$} 
& \multicolumn{4}{c}{Case 1} 
& \multicolumn{4}{c}{Case 2} 
& \multicolumn{4}{c}{Case 3} \\
\cmidrule(lr){4-7} \cmidrule(lr){8-11} \cmidrule(lr){12-15}
 &  &  & PRO & SUB & SRS & LTM 
    & PRO & SUB & SRS & LTM 
    & PRO & SUB & SRS & LTM \\
\midrule
\multicolumn{15}{c}{$G(x)=x$} \\ 
Case-cohort & 0.1 & 2000 & 0.382 & 0.632 & 0.544 & {\bf 0.343} & {\bf 0.388} & 0.524 & 0.480 & 0.424 & {\bf 0.464} & 0.602 & 0.548 & 0.604 \\
study       &     &      & (0.106) & (0.200) & (0.178) & (0.116) & (0.074) & (0.118) & (0.104) & (0.076) & (0.086) & (0.140) & (0.122) & (0.110) \\
            &     & 3000 & 0.319 & 0.516 & 0.468 & {\bf 0.271} & {\bf 0.319} & 0.410 & 0.381 & 0.388 & {\bf 0.402} & 0.487 & 0.460 & 0.577 \\
            &     &      & (0.076) & (0.153) & (0.132) & (0.088) & (0.046) & (0.065) & (0.062) & (0.054) & (0.060) & (0.087) & (0.074) & (0.082) \\
            
            & 0.2 & 2000 & 0.292 & 0.443 & 0.352 & {\bf 0.273} & {\bf 0.317} & 0.420 & 0.353 & 0.378 & {\bf 0.374} & 0.444 & 0.398 & 0.517 \\
            &     &      & (0.075) & (0.133) & (0.097) & (0.088) & (0.055) & (0.088) & (0.068) & (0.053) & (0.053) & (0.068) & (0.056) & (0.064) \\
            &     & 3000 & 0.269 & 0.411 & 0.319 & {\bf 0.223} & {\bf 0.292} & 0.360 & 0.319 & 0.355 & {\bf 0.340} & 0.434 & 0.381 & 0.500 \\
            &     &      & (0.061) & (0.108) & (0.080) & (0.072) & (0.038) & (0.052) & (0.041) & (0.041) & (0.048) & (0.070) & (0.066) & (0.053) \\
            
Generalized & 0.2 & 2000 & 0.328 & 0.450 & 0.393 & {\bf 0.303} & {\bf 0.339} & 0.422 & 0.384 & 0.384 & {\bf 0.391} & 0.447 & 0.417 & 0.525 \\
case-cohort &     &      & (0.088) & (0.135) & (0.114) & (0.101) & (0.061) & (0.094) & (0.079) & (0.057) & (0.058) & (0.075) & (0.066) & (0.073) \\
study       &     & 3000 & 0.286 & 0.394 & 0.338 & {\bf 0.247} & {\bf 0.304} & 0.366 & 0.331 & 0.363 & {\bf 0.359} & 0.442 & 0.399 & 0.506 \\
            &     &      & (0.071) & (0.111) & (0.092) & (0.080) & (0.047) & (0.064) & (0.058) & (0.044) & (0.054) & (0.080) & (0.067) & (0.057) \\
            
            & 0.3 & 2000 & 0.301 & 0.385 & 0.324 & {\bf 0.263} & {\bf 0.321} & 0.391 & 0.335 & 0.363 & {\bf 0.378} & 0.454 & 0.389 & 0.494 \\
            &     &      & (0.078) & (0.110) & (0.089) & (0.089) & (0.058) & (0.079) & (0.061) & (0.049) & (0.058) & (0.083) & (0.061) & (0.056) \\
            &     & 3000 & 0.263 & 0.331 & 0.272 & {\bf 0.216} & {\bf 0.283} & 0.343 & 0.296 & 0.346 & {\bf 0.343} & 0.423 & 0.360 & 0.475 \\
            &     &      & (0.059) & (0.082) & (0.065) & (0.068) & (0.042) & (0.060) & (0.046) & (0.037) & (0.049) & (0.075) & (0.058) & (0.040) \\
\multicolumn{15}{c}{$G(x)=2\log(1 + x/2)$} \\ 
Case-cohort & 0.1 & 2000 & 0.371 & 0.645 & 0.550 & {\bf 0.342} & {\bf 0.389} & 0.582 & 0.510 & 0.416 & {\bf 0.449} & 0.631 & 0.563 & 0.583 \\
study       &     &      & (0.102) & (0.216) & (0.180) & (0.117) & (0.076) & (0.138) & (0.116) & (0.071) & (0.082) & (0.154) & (0.130) & (0.103) \\
            &     & 3000 & 0.315 & 0.533 & 0.473 & {\bf 0.281} & {\bf 0.330} & 0.458 & 0.416 & 0.384 & {\bf 0.417} & 0.581 & 0.533 & 0.555 \\
            &     &      & (0.082) & (0.146) & (0.125) & (0.089) & (0.054) & (0.094) & (0.084) & (0.053) & (0.074) & (0.132) & (0.112) & (0.077) \\
            
            & 0.2 & 2000 & 0.312 & 0.471 & 0.370 & {\bf 0.288} & {\bf 0.332} & 0.438 & 0.376 & 0.378 & {\bf 0.397} & 0.526 & 0.436 & 0.509 \\
            &     &      & (0.079) & (0.142) & (0.100) & (0.093) & (0.055) & (0.088) & (0.067) & (0.055) & (0.063) & (0.105) & (0.073) & (0.061) \\
            &     & 3000 & 0.256 & 0.383 & 0.302 & {\bf 0.231} & {\bf 0.298} & 0.397 & 0.329 & 0.355 & {\bf 0.359} & 0.483 & 0.403 & 0.493 \\
            &     &      & (0.060) & (0.106) & (0.079) & (0.072) & (0.046) & (0.078) & (0.058) & (0.041) & (0.052) & (0.086) & (0.066) & (0.049) \\
            
Generalized & 0.2 & 2000 & 0.339 & 0.460 & 0.395 & {\bf 0.323} & {\bf 0.361} & 0.408 & 0.379 & 0.390 & {\bf 0.428} & 0.520 & 0.464 & 0.524 \\
case-cohort &     &      & (0.091) & (0.134) & (0.113) & (0.104) & (0.050) & (0.068) & (0.061) & (0.062) & (0.075) & (0.107) & (0.082) & (0.073) \\
study       &     & 3000 & 0.288 & 0.390 & 0.340 & {\bf 0.265} & {\bf 0.312} & 0.381 & 0.347 & 0.366 & {\bf 0.390} & 0.496 & 0.445 & 0.496 \\
            &     &      & (0.074) & (0.114) & (0.097) & (0.083) & (0.050) & (0.070) & (0.063) & (0.047) & (0.065) & (0.100) & (0.082) & (0.051) \\
            
            & 0.3 & 2000 & 0.319 & 0.414 & 0.339 & {\bf 0.290} & {\bf 0.323} & 0.390 & 0.337 & 0.367 & {\bf 0.442} & 0.558 & 0.457 & 0.496 \\
            &     &      & (0.077) & (0.114) & (0.088) & (0.092) & (0.056) & (0.076) & (0.060) & (0.050) & (0.083) & (0.121) & (0.086) & (0.058) \\
            &     & 3000 & 0.287 & 0.372 & 0.306 & {\bf 0.233} & {\bf 0.296} & 0.357 & 0.309 & 0.349 & {\bf 0.381} & 0.479 & 0.403 & 0.474 \\
            &     &      & (0.064) & (0.094) & (0.072) & (0.072) & (0.047) & (0.064) & (0.050) & (0.040) & (0.058) & (0.084) & (0.066) & (0.041) \\
\multicolumn{15}{c}{$G(x)=\log(1 + x)$} \\ 
Case-cohort & 0.1 & 2000 & 0.386 & 0.661 & 0.564 & {\bf 0.359} & {\bf 0.371} & 0.463 & 0.432 & 0.413 & {\bf 0.449} & 0.629 & 0.580 & 0.566 \\
study       &     &      & (0.103) & (0.220) & (0.175) & (0.117) & (0.059) & (0.092) & (0.076) & (0.073) & (0.079) & (0.151) & (0.128) & (0.091) \\
            &     & 3000 & 0.318 & 0.534 & 0.461 & {\bf 0.287} & {\bf 0.344} & 0.510 & 0.462 & 0.384 & {\bf 0.401} & 0.569 & 0.510 & 0.538 \\
            &     &      & (0.080) & (0.164) & (0.129) & (0.097) & (0.062) & (0.113) & (0.101) & (0.054) & (0.063) & (0.122) & (0.104) & (0.071) \\
            
            & 0.2 & 2000 & 0.320 & 0.486 & 0.389 & {\bf 0.298} & {\bf 0.338} & 0.461 & 0.382 & 0.380 & {\bf 0.405} & 0.530 & 0.451 & 0.509 \\
            &     &      & (0.080) & (0.138) & (0.107) & (0.096) & (0.059) & (0.093) & (0.072) & (0.058) & (0.068) & (0.110) & (0.080) & (0.066) \\
            &     & 3000 & 0.260 & 0.399 & 0.320 & {\bf 0.238} & {\bf 0.307} & 0.421 & 0.351 & 0.356 & {\bf 0.366} & 0.473 & 0.401 & 0.488 \\
            &     &      & (0.063) & (0.113) & (0.085) & (0.076) & (0.049) & (0.083) & (0.061) & (0.043) & (0.051) & (0.083) & (0.061) & (0.050) \\
            
Generalized & 0.2 & 2000 & 0.358 & 0.492 & 0.423 & {\bf 0.338} & {\bf 0.366} & 0.447 & 0.404 & 0.396 & {\bf 0.436} & 0.534 & 0.485 & 0.517 \\
case-cohort &     &      & (0.095) & (0.147) & (0.124) & (0.110) & (0.066) & (0.097) & (0.076) & (0.064) & (0.076) & (0.100) & (0.088) & (0.068) \\
study       &     & 3000 & 0.304 & 0.413 & 0.357 & {\bf 0.272} & {\bf 0.332} & 0.379 & 0.351 & 0.364 & {\bf 0.386} & 0.473 & 0.433 & 0.496 \\
            &     &      & (0.076) & (0.116) & (0.097) & (0.086) & (0.045) & (0.058) & (0.050) & (0.047) & (0.057) & (0.079) & (0.068) & (0.051) \\
            
            & 0.3 & 2000 & 0.346 & 0.454 & 0.373 & {\bf 0.303} & {\bf 0.349} & 0.436 & 0.368 & 0.376 & {\bf 0.417} & 0.496 & 0.436 & 0.500 \\
            &     &      & (0.083) & (0.125) & (0.094) & (0.094) & (0.061) & (0.087) & (0.068) & (0.056) & (0.066) & (0.093) & (0.073) & (0.060) \\
            &     & 3000 & 0.287 & 0.365 & 0.312 & {\bf 0.248} & {\bf 0.307} & 0.376 & 0.322 & 0.354 & {\bf 0.374} & 0.448 & 0.389 & 0.479 \\
            &     &      & (0.067) & (0.098) & (0.079) & (0.078) & (0.048) & (0.066) & (0.050) & (0.041) & (0.053) & (0.075) & (0.056) & (0.041) \\
\bottomrule
\end{tabular}
}
\end{table}

Table~\ref{table1} presents the performance comparison of four methods in estimating the covariate-dependent function under different simulation settings. For each replication, the relative error (RE) was computed on the test set, and the table reports the mean and standard deviation of RE across $1000$ replications. When the covariate-dependent function follows Case~1, the perfectly specified LTM method achieves the best performance, while the proposed method is only slightly inferior. Under Case~2 and the more complex Case~3, the proposed method clearly outperforms the others, achieving the smallest relative error. It is noteworthy that across all simulation settings, the proposed method consistently exhibits smaller RE than both the subcohort-based method and the method based on a simple random sample of the same size as the (generalized) case-cohort sample. The advantage is particularly pronounced when the event rate is $p_e=0.1$. Moreover, as the cohort size increases from $1800$ to $2700$ (corresponding to the sample size $n$ increasing from $2000$ to $3000$), the RE of the proposed estimator decreases. This phenomenon is theoretically supported by Theorems~\ref{th1} and~\ref{th2}.

\begin{figure}[htbp]
    \centering

    \begin{subfigure}{0.48\linewidth}
        \centering
        \includegraphics[width=\linewidth]{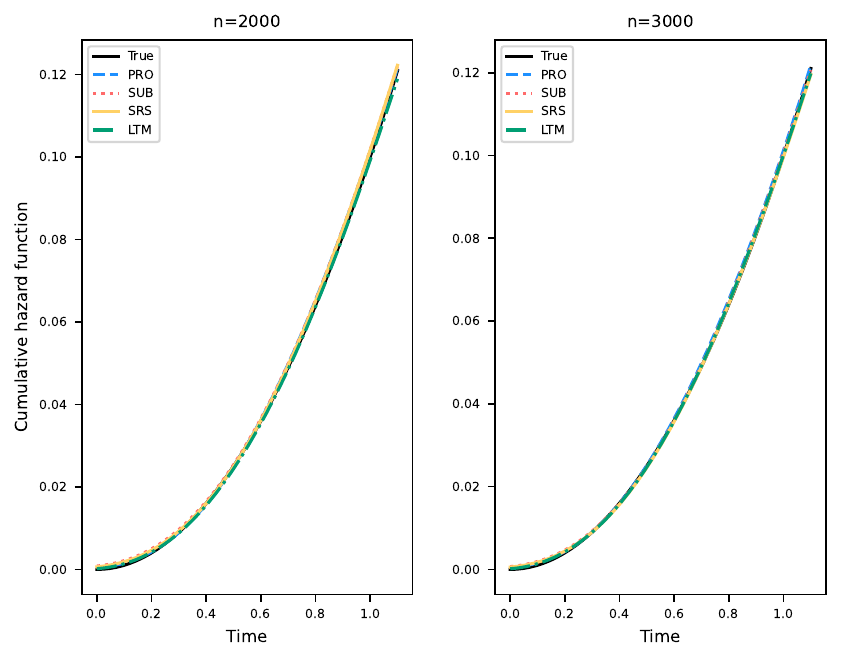}
        \caption{$p_e=0.1$,Case 1}
        \label{fig:sub_a}
    \end{subfigure}
    \hfill
    \begin{subfigure}{0.48\linewidth}
        \centering
        \includegraphics[width=\linewidth]{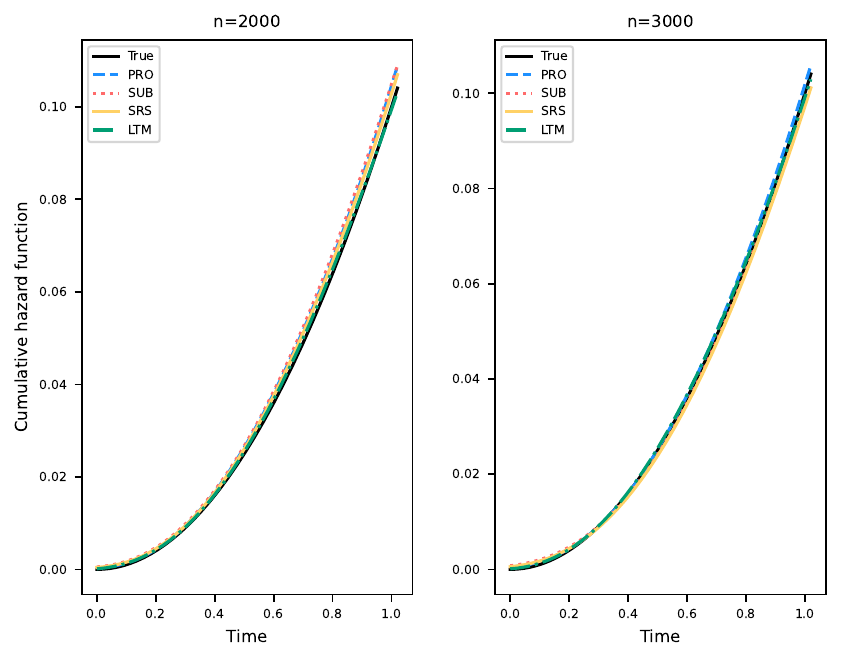}
        \caption{$p_e=0.1$,Case 2}
        \label{fig:sub_b}
    \end{subfigure}
    
    \begin{subfigure}{0.48\linewidth}
        \centering
        \includegraphics[width=\linewidth]{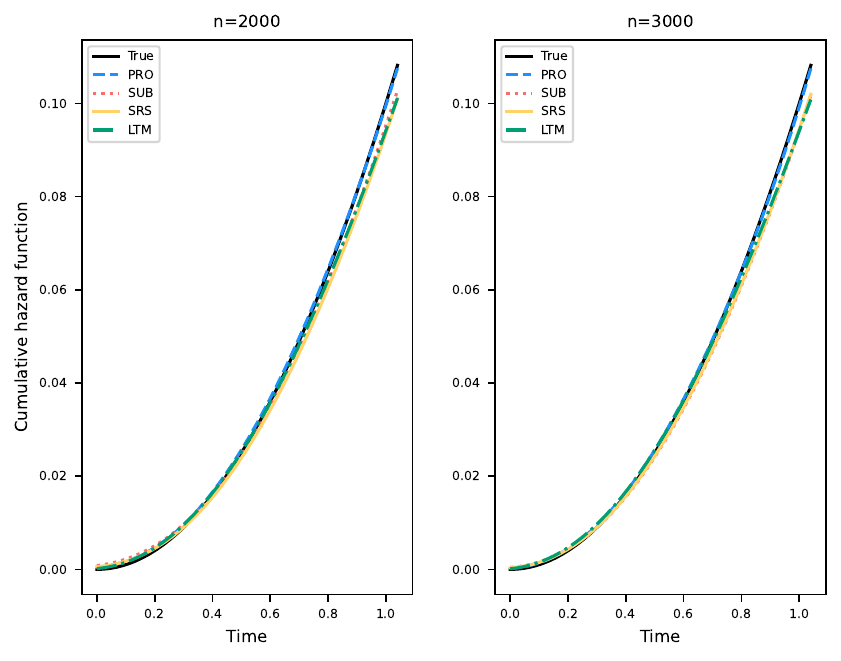}
        \caption{$p_e=0.1$,Case 3}
        \label{fig:sub_c}
    \end{subfigure}

    \caption{Estimates of $\Lambda(\cdot)$ by different methods based on 1000 replications}
    \label{figure1}
\end{figure}

In the simulation, the expectation of the true covariate-dependent function $g$ was set to $0$, and the solution of maximizing the weighted log-likelihood is only unique up to a constant. Therefore, $\hat{g}$ was centered centered to have zero mean by subtracting its sample mean $\bar{\hat{g}}$ computed on the test set. Accordingly, to preserve model equivalence, the estimate of the cumulative baseline hazard function was adjusted by multiplying it by $\exp(\bar{\hat{g}})$. Figure~\ref{figure1} presents a comparison of the estimated cumulative baseline hazard functions, obtained by the four methods, against the true function for the simulation setting with $G(x) = \log(1 + x)$ and an event rate of $p_e = 0.1$. As shown in the figure, when the covariate-dependent function corresponds to Case 1 or Case 2, all four methods provide accurate estimates of the cumulative baseline hazard function. However, under the more complex Case 3, our method exhibits the smallest bias and is able to converge to the true cumulative baseline hazard function.

\begin{table}[htbp]
\centering
\caption{\label{table2}Mean squared prediction errors (×1000) and standard deviations for different methods}
\resizebox{6.6in}{!}{
\begin{tabular}{ccccccccccccccc}
\toprule
\multirow{2}{*}{} & \multirow{2}{*}{$p_{\text{e}}$} & \multirow{2}{*}{$n$} 
& \multicolumn{4}{c}{Case 1} 
& \multicolumn{4}{c}{Case 2} 
& \multicolumn{4}{c}{Case 3} \\
\cmidrule(lr){4-7} \cmidrule(lr){8-11} \cmidrule(lr){12-15}
 &  &  & PRO & SUB & SRS & LTM 
    & PRO & SUB & SRS & LTM 
    & PRO & SUB & SRS & LTM \\
\midrule
\multicolumn{15}{c}{$G(x)=x$} \\ 
Case-cohort & 0.1 & 2000 & {\bf 0.198} & 0.503 & 0.379 & 0.199 & {\bf 0.725} & 1.241 & 1.002 & 0.840 & {\bf 1.569} & 2.654 & 2.000 & 1.821 \\
study       &     &      & (0.120) & (0.263) & (0.213) & (0.182) & (0.369) & (0.615) & (0.464) & (0.424) & (0.744) & (1.412) & (1.027) & (0.832) \\
            &     & 3000 & 0.149 & 0.355 & 0.285 & {\bf 0.118} & {\bf 0.553} & 0.917 & 0.760 & 0.688 & {\bf 1.265} & 1.993 & 1.723 & 1.666 \\
            &     &      & (0.100) & (0.193) & (0.166) & (0.085) & (0.266) & (0.417) & (0.336) & (0.285) & (0.522) & (0.968) & (0.785) & (0.636) \\
            
            & 0.2 & 2000 & 0.414 & 0.916 & 0.562 & {\bf 0.405} & {\bf 1.271} & 2.087 & 1.400 & 1.604 & {\bf 2.070} & 3.062 & 2.355 & 2.925 \\
            &     &      & (0.220) & (0.480) & (0.299) & (0.301) & (0.499) & (0.865) & (0.517) & (0.500) & (0.755) & (1.200) & (0.881) & (0.937) \\
            &     & 3000 & 0.333 & 0.700 & 0.426 & {\bf 0.259} & {\bf 1.059} & 1.609 & 1.218 & 1.419 & {\bf 1.531} & 2.506 & 1.611 & 2.700 \\
            &     &      & (0.174) & (0.346) & (0.192) & (0.188) & (0.315) & (0.531) & (0.371) & (0.362) & (0.509) & (0.844) & (0.541) & (0.714) \\
            
Generalized & 0.2 & 2000 & 0.503 & 0.914 & 0.700 & {\bf 0.481} & {\bf 1.405} & 2.076 & 1.641 & 1.721 & {\bf 2.303} & 2.986 & 2.552 & 3.054 \\
case-cohort &     &      & (0.274) & (0.460) & (0.380) & (0.340) & (0.553) & (0.922) & (0.646) & (0.572) & (0.863) & (1.223) & (1.062) & (0.995) \\
study       &     & 3000 & 0.371 & 0.673 & 0.507 & {\bf 0.322} & {\bf 1.109} & 1.557 & 1.247 & 1.480 & {\bf 1.748} & 2.422 & 2.002 & 2.771 \\
            &     &      & (0.184) & (0.342) & (0.270) & (0.218) & (0.378) & (0.571) & (0.412) & (0.416) & (0.583) & (0.841) & (0.689) & (0.756) \\
            
            & 0.3 & 2000 & 0.804 & 1.299 & 0.880 & {\bf 0.680} & {\bf 1.989} & 2.852 & 2.039 & 2.521 & {\bf 2.527} & 3.497 & 2.643 & 3.917 \\
            &     &      & (0.410) & (0.671) & (0.433) & (0.439) & (0.786) & (1.083) & (0.693) & (0.777) & (0.848) & (1.289) & (0.892) & (1.019) \\
            &     & 3000 & 0.600 & 0.932 & 0.610 & {\bf 0.460} & {\bf 1.525} & 2.173 & 1.620 & 2.242 & {\bf 2.182} & 3.020 & 2.245 & 3.644 \\
            &     &      & (0.285) & (0.437) & (0.276) & (0.276) & (0.467) & (0.709) & (0.464) & (0.508) & (0.684) & (0.993) & (0.704) & (0.787) \\
\multicolumn{15}{c}{$G(x)=2\log(1 + x/2)$} \\ 
Case-cohort & 0.1 & 2000 & {\bf 0.183} & 0.478 & 0.354 & 0.183 & 0.747 & 1.282 & 1.016 & {\bf 0.745} & {\bf 1.237} & 2.031 & 1.725 & 1.566 \\
study       &     &      & (0.112) & (0.260) & (0.189) & (0.146) & (0.426) & (0.632) & (0.456) & (0.362) & (0.617) & (1.093) & (0.932) & (0.751) \\
            &     & 3000 & 0.130 & 0.363 & 0.283 & {\bf 0.121} & {\bf 0.496} & 0.880 & 0.717 & 0.603 & {\bf 0.938} & 1.623 & 1.336 & 1.474 \\
            &     &      & (0.068) & (0.177) & (0.143) & (0.092) & (0.205) & (0.406) & (0.307) & (0.217) & (0.412) & (0.766) & (0.618) & (0.619) \\

            & 0.2 & 2000 & 0.418 & 0.942 & 0.584 & {\bf 0.386} & {\bf 1.183} & 1.980 & 1.416 & 1.442 & {\bf 1.831} & 2.966 & 2.080 & 2.652 \\
            &     &      & (0.218) & (0.512) & (0.305) & (0.268) & (0.442) & (0.821) & (0.538) & (0.481) & (0.713) & (1.268) & (0.796) & (0.883) \\
            &     & 3000 & 0.291 & 0.644 & 0.405 & {\bf 0.253} & {\bf 0.992} & 1.586 & 1.104 & 1.271 & {\bf 1.457} & 2.357 & 1.677 & 2.439 \\
            &     &      & (0.150) & (0.339) & (0.210) & (0.165) & (0.339) & (0.564) & (0.354) & (0.335) & (0.499) & (0.830) & (0.565) & (0.698) \\

Generalized & 0.2 & 2000 & 0.523 & 0.936 & 0.683 & {\bf 0.497} & {\bf 1.404} & 1.879 & 1.590 & 1.556 & {\bf 2.010} & 2.934 & 2.374 & 2.782 \\
case-cohort &     &      & (0.298) & (0.508) & (0.359) & (0.326) & (0.499) & (0.751) & (0.634) & (0.548) & (0.789) & (1.227) & (0.988) & (0.975) \\
study       &     & 3000 & 0.356 & 0.643 & 0.499 & {\bf 0.322} & {\bf 1.036} & 1.472 & 1.205 & 1.360 & {\bf 1.655} & 2.409 & 1.955 & 2.511 \\
            &     &      & (0.178) & (0.326) & (0.271) & (0.202) & (0.360) & (0.537) & (0.411) & (0.384) & (0.587) & (0.901) & (0.699) & (0.695) \\

            & 0.3 & 2000 & 0.827 & 1.384 & 0.891 & {\bf 0.730} & {\bf 1.788} & 2.580 & 1.878 & 2.279 & 2.871 & 4.085 & {\bf 2.866} & 3.577 \\
            &     &      & (0.424) & (0.710) & (0.453) & (0.464) & (0.662) & (1.018) & (0.674) & (0.691) & (1.062) & (1.499) & (0.972) & (1.019) \\
            &     & 3000 & 0.628 & 1.004 & 0.679 & {\bf 0.459} & {\bf 1.461} & 2.088 & 1.548 & 2.009 & {\bf 2.103} & 3.075 & 2.226 & 3.263 \\
            &     &      & (0.300) & (0.469) & (0.311) & (0.282) & (0.462) & (0.670) & (0.463) & (0.448) & (0.676) & (0.954) & (0.673) & (0.770) \\
\multicolumn{15}{c}{$G(x)=\log(1 + x)$} \\ 
Case-cohort & 0.1 & 2000 & {\bf 0.177} & 0.495 & 0.359 & 0.177 & {\bf 0.583} & 0.993 & 0.814 & 0.685 & {\bf 1.148} & 1.946 & 1.626 & 1.463 \\
study       &     &      & (0.101) & (0.268) & (0.192) & (0.124) & (0.274) & (0.550) & (0.385) & (0.333) & (0.558) & (1.054) & (0.890) & (0.705) \\
            &     & 3000 & 0.125 & 0.332 & 0.246 & {\bf 0.119} & {\bf 0.509} & 0.937 & 0.756 & 0.551 & {\bf 0.921} & 1.537 & 1.246 & 1.328 \\
            &     &      & (0.068) & (0.173) & (0.124) & (0.087) & (0.237) & (0.410) & (0.316) & (0.195) & (0.403) & (0.736) & (0.558) & (0.533) \\
            
            & 0.2 & 2000 & 0.413 & 0.931 & 0.580 & {\bf 0.387} & {\bf 1.135} & 1.969 & 1.317 & 1.343 & {\bf 1.704} & 2.793 & 1.994 & 2.364 \\
            &     &      & (0.217) & (0.494) & (0.301) & (0.256) & (0.451) & (0.806) & (0.497) & (0.453) & (0.695) & (1.232) & (0.785) & (0.854) \\
            &     & 3000 & 0.280 & 0.666 & 0.417 & {\bf 0.245} & {\bf 0.917} & 1.579 & 1.100 & 1.165 & {\bf 1.376} & 2.232 & 1.608 & 2.237 \\
            &     &      & (0.140) & (0.349) & (0.219) & (0.158) & (0.325) & (0.607) & (0.373) & (0.326) & (0.460) & (0.827) & (0.535) & (0.656) \\
            
Generalized & 0.2 & 2000 & 0.508 & 0.933 & 0.689 & {\bf 0.485} & {\bf 1.254} & 1.840 & 1.501 & 1.462 & {\bf 1.904} & 2.807 & 2.299 & 2.510 \\
case-cohort &     &      & (0.271) & (0.498) & (0.352) & (0.311) & (0.474) & (0.814) & (0.595) & (0.491) & (0.754) & (1.155) & (0.937) & (0.886) \\
study       &     & 3000 & 0.359 & 0.653 & 0.485 & {\bf 0.320} & {\bf 1.032} & 1.378 & 1.170 & 1.245 & {\bf 1.503} & 2.218 & 1.817 & 2.268 \\
            &     &      & (0.178) & (0.328) & (0.249) & (0.193) & (0.330) & (0.460) & (0.365) & (0.362) & (0.532) & (0.852) & (0.629) & (0.674) \\
            
            & 0.3 & 2000 & 0.830 & 1.420 & 0.952 & {\bf 0.686} & {\bf 1.822} & 2.755 & 1.956 & 2.120 & {\bf 2.381} & 3.381 & 2.572 & 3.284 \\
            &     &      & (0.405) & (0.705) & (0.469) & (0.417) & (0.675) & (1.041) & (0.704) & (0.661) & (0.845) & (1.247) & (0.892) & (0.904) \\
            &     & 3000 & 0.573 & 0.942 & 0.655 & {\bf 0.468} & {\bf 1.368} & 2.037 & 1.476 & 1.838 & {\bf 1.908} & 2.684 & 2.038 & 2.965 \\
            &     &      & (0.276) & (0.470) & (0.320) & (0.284) & (0.433) & (0.696) & (0.456) & (0.453) & (0.600) & (0.892) & (0.642) & (0.681) \\
\bottomrule
\end{tabular}
}
\end{table}

Table~\ref{table2} compares the mean squared prediction error (MSPE) between the estimated survival functions and the true survival function across the four methods. The results indicate that when the covariate-dependent function corresponds to Case 1, the LTM method achieves the best performance due to its correct model specification, while the proposed deep learning approach based on the (generalized) case-cohort sample delivers the second-best performance, close to that of the LTM method. When the covariate-dependent function is more complex, the proposed method attains the smallest MSPE in the majority of settings, outperforming the deep learning approaches based on subcohort or simple random samples. Moreover, as the cohort size increases, the MSPE of the proposed method decreases, as expected.

These simulation results indicate that the LTM method relies on strong model assumptions and is less suitable for complex real-world data. In contrast, DNN-based approaches offer greater flexibility and provide a more robust option for modeling covariate effects. Notably, our method demonstrates higher stability and smaller prediction errors than SUB and SRS, particularly in settings with low event rates.

\section{An Application}
\label{applic}

In this section, the proposed method is applied to a set of interval-censored data on child mortality obtained from the 2003 Nigeria Demographic and Health Survey \citep{Kneib2006Mixed}. In this study, if a child died within the first two months after birth, the exact time of death could be recorded. For deaths occurring after this period, information was collected through interviews with the mothers, which resulted in interval-censored data on the death time. Six covariates were included in the analysis: the mother’s age at birth (AGE) and body mass index (BMI), both continuous covariates and standardized in subsequent analyses; and four binary covariates: whether the child was born in a hospital (HOSP, 1 for hospital birth and 0 otherwise), the child’s gender (GENDER, 1 for male and 0 for female), whether the mother had received higher education (EDU, 1 if yes, 0 otherwise), and whether the household resided in an urban area (URBAN, 1 for urban, 0 otherwise).

Our analysis included $5730$ children, among whom $663$ had interval-censored observations for the death time, while the remaining observations were right-censored. The event rate was $0.1157$. To assess model performance, ten-fold cross-validation was employed. In each round, one fold was retained as the test set, and the remaining nine folds were used as the training-validation set, which was further randomly divided into the training and validation set in a $9{:}1$ ratio. The training set corresponded to the entire study cohort, and a case-cohort sample was artificially constructed by selecting a subcohort from the cohort using Bernoulli sampling with probability $0.2$. We considered the model~\eqref{equa1}, which was fitted using the proposed method based on the case-cohort sample. In addition to the proposed method, the three approaches introduced in Section~\ref{simul} were also used for comparison.

Hyperparameter tuning was performed through grid search, using the negative log-likelihood on the validation set as the evaluation criterion, and early stopping was employed to prevent overfitting based on the validation performance. Compared with the hyperparameters considered in the simulation study, the degree $m$ of the Bernstein polynomials and the parameter $r$ of the logarithmic transformation function $G$ were additionally included in the real-data analysis. For the linear transformation model based on the case-cohort sample, both $m$ and $r$ were specified as hyperparameters. The interval $[c,u]$ in the Bernstein polynomials was set to the minimum and the maximum finite values of the observed ${L_i,R_i}$. The hyperparameter search space was defined as follows:
\begin{itemize}
    \setlength{\itemsep}{0pt}
    \setlength{\parskip}{0pt}
    \item Degree $m$ of the Bernstein polynomials: $6$, $7$, $8$;
    \item Parameter $r$ of the logarithmic transformation function $G$: $0$, $0.5$, $1$;
    \item Batch size: $32$, $64$;
    \item Number of hidden layers $H$: $1$, $2$, $3$;
    \item Number of neurons per hidden layer $p_h$: $50$, $100$, $200$, $300$;
    \item Dropout rate: $0$, $0.1$, $0.3$;
    \item Learning rate for BPNet: $0.01$, $0.005$;
    \item Learning rate for the covariate network: $1 \times 10^{-4}$, $5 \times 10^{-5}$.
\end{itemize}

In the analysis of the real dataset, the predictive accuracy of each method was evaluated using the integrated Brier score (IBS).
\begin{equation*}
\text{IBS}(\widehat{S}) = \frac{1}{n_1}\sum_{i=1}^{n_1} \frac{1}{u-c} \int_c^u \left\{ I(T_i > t \mid Z_i) - \widehat{S}(t \mid Z_i) \right\}^2 dt,
\end{equation*}
where $n_1$ denotes the size of the test set and $[c,u]$ is taken to be the same as in the Bernstein polynomials. For an individual $i$, $I(T_i > t \mid Z_i) = 0$ if $R_i < t$, and $I(T_i > t \mid Z_i) = 1$ if $L_i \geq t$. When $L_i < t \leq R_i$, the true value of $I(T_i > t \mid Z_i)$ is unknown and is estimated by $\hat{I}(T_i > t \mid Z_i) = \frac{\hat{S}(t \mid Z_i) - \hat{S}(R_i \mid Z_i)}{\hat{S}(L_i \mid Z_i) - \hat{S}(R_i \mid Z_i)}$. In the special case where $L_i < t \leq R_i = \infty$, it is estimated as $\hat{I}(T_i > t \mid Z_i) = \frac{\hat{S}(t \mid Z_i)}{\hat{S}(L_i \mid Z_i)}$.

\begin{table}[htbp]
\centering
\caption{\label{table3}Integrated Brier scores (×100) and standard deviations for different methods}
\begin{tabular}{ccccc}
\toprule
\textbf{Model} & \textbf{PRO} & \textbf{SUB} & \textbf{SRS} & \textbf{LTM}\\
\midrule
IBS & {\bf 8.249} & 8.307 & 8.301 & 8.262\\
SD & (0.266) & (0.307) & (0.271) & (0.237) \\
 \bottomrule
\end{tabular}
\end{table}

Table~\ref{table3} presents a comparison of the predictive performance of four methods. The reported IBS and SD correspond to the mean and standard deviation of the IBS, respectively, calculated on the test sets through ten-fold cross-validation, with all values multiplied by $100$ for presentation. A smaller IBS indicates better predictive performance. The results show that our method achieves the best predictive accuracy. In addition, the proposed method outperforms SUB and SRS in both accuracy and stability.

After fitting the model \eqref{equa1} using the case-cohort sample, we applied the SHAP method to provide an interpretable analysis of the covariate neural network’s predictions on the test set. The SHAP method decomposes the model prediction for each sample into additive contributions of individual covariates relative to a baseline prediction, with each contribution being called the SHAP value \citep{Lundberg2017unified}. Let $Z_i$ denote the covariates of the $i$th sample in the test set, and let $\hat{g}(Z_i)$ denote the corresponding prediction of the covariate neural network. The SHAP values satisfy the following additive relationship $\hat{g}(Z_i) = g_{\text{base}} + \sum_j \phi_{i,j},$ where $g_{\text{base}}$ represents the expected prediction of the covariate neural network over the background data, which was obtained by randomly sampling from the case-cohort sample according to the original case proportion. Here, $\phi_{i,j}$ is the SHAP value of the $j$th covariate for the $i$th sample, quantifying the specific contribution of this covariate to the prediction relative to the baseline $g_{\text{base}}$.

\begin{figure}[htbp]
    \centering

    \begin{subfigure}{0.48\linewidth}
        \centering
        \includegraphics[width=\linewidth]{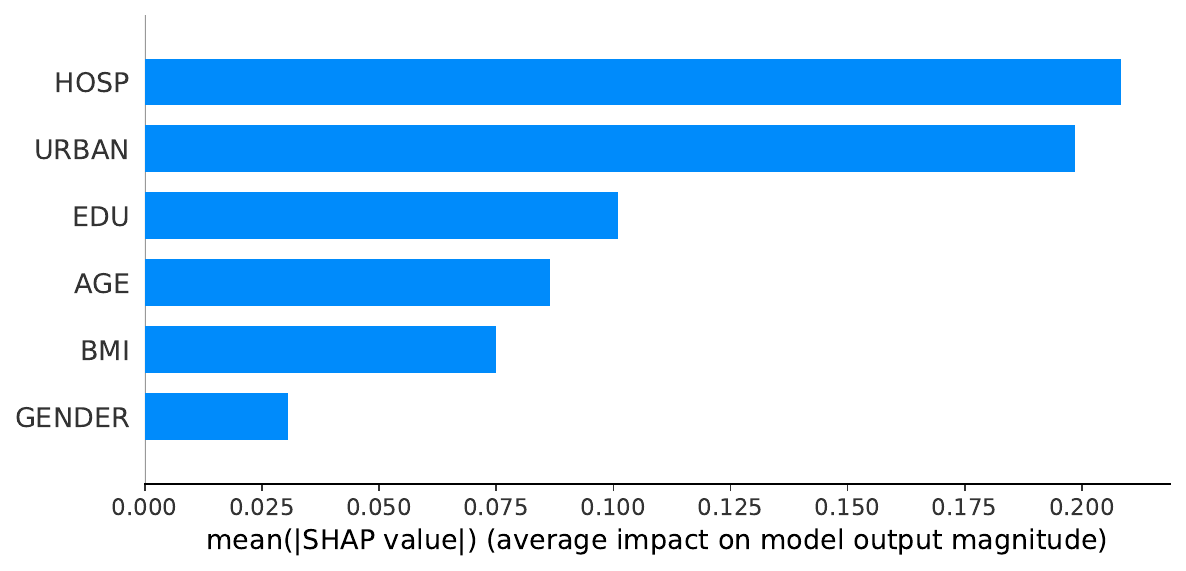}
        \caption{Mean absolute SHAP values for each covariate}
        \label{fig2:shap_bar}
    \end{subfigure}
    \hfill
    \begin{subfigure}{0.48\linewidth}
        \centering
        \includegraphics[width=\linewidth]{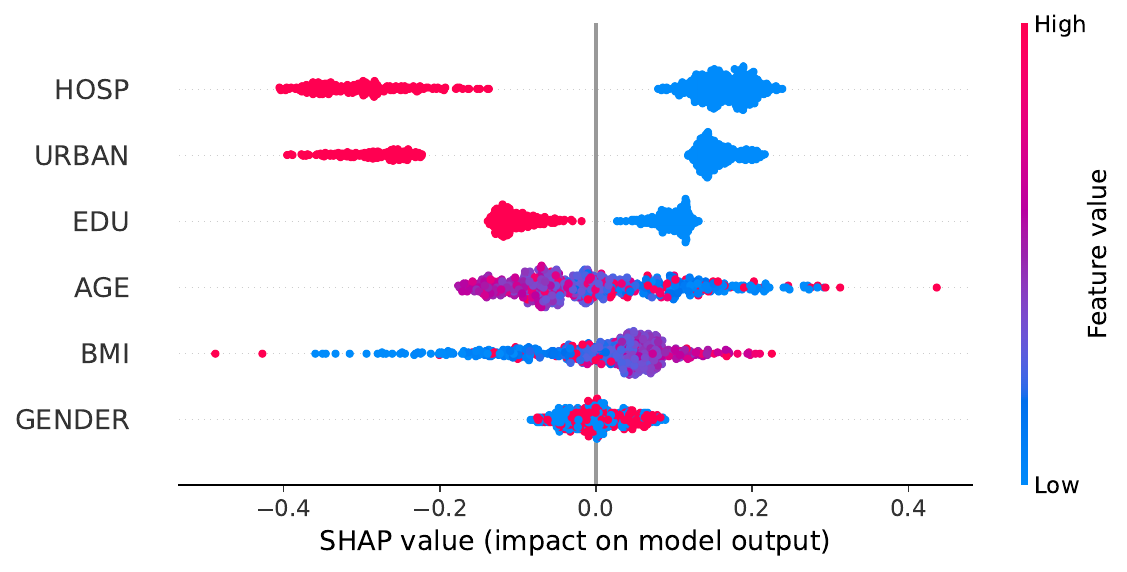}
        \caption{Distribution of SHAP values for each covariate}
        \label{fig2:shap_dot}
    \end{subfigure}

    \caption{SHAP analysis of the neural network used to estimate the covariate-dependent function}
    \label{fig2:shap_bar_dot}
\end{figure}

Figure~\ref{fig2:shap_bar_dot} presents the SHAP analysis of the neural network used to approximate the covariate-dependent function in the model, based on one fold of ten-fold cross-validation. Figure~\ref{fig2:shap_bar} shows a bar plot of the mean absolute SHAP values for each covariate, reflecting their relative importance in predicting the network output $\hat{g}(Z)$. In this fold, HOSP has the largest impact on model predictions, followed by URBAN, then EDU, AGE, and BMI, while GENDER has a relatively minor effect. Figure~\ref{fig2:shap_dot} shows the distribution of SHAP values, illustrating the direction and magnitude of each covariate’s effect on the predictions. The color of each point indicates the covariate value (red for higher values, blue for lower values), and the x-axis represents the SHAP values for the covariates for the corresponding samples. It can be observed that when HOSP, URBAN, and EDU take the value $1$, they correspond to negative SHAP values, indicating that hospital delivery, urban residence, and higher education tend to decrease the neural network output $\hat{g}(Z)$ relative to the baseline $g_{\text{base}}$, corresponding to lower cumulative hazard function and higher survival probability at fixed times. The effects of AGE, BMI, and GENDER are more balanced, showing no obvious result can be seen in the figure.

\begin{figure}[htbp]
    \centering

    \begin{subfigure}{0.48\linewidth}
        \centering
        \includegraphics[width=\linewidth]{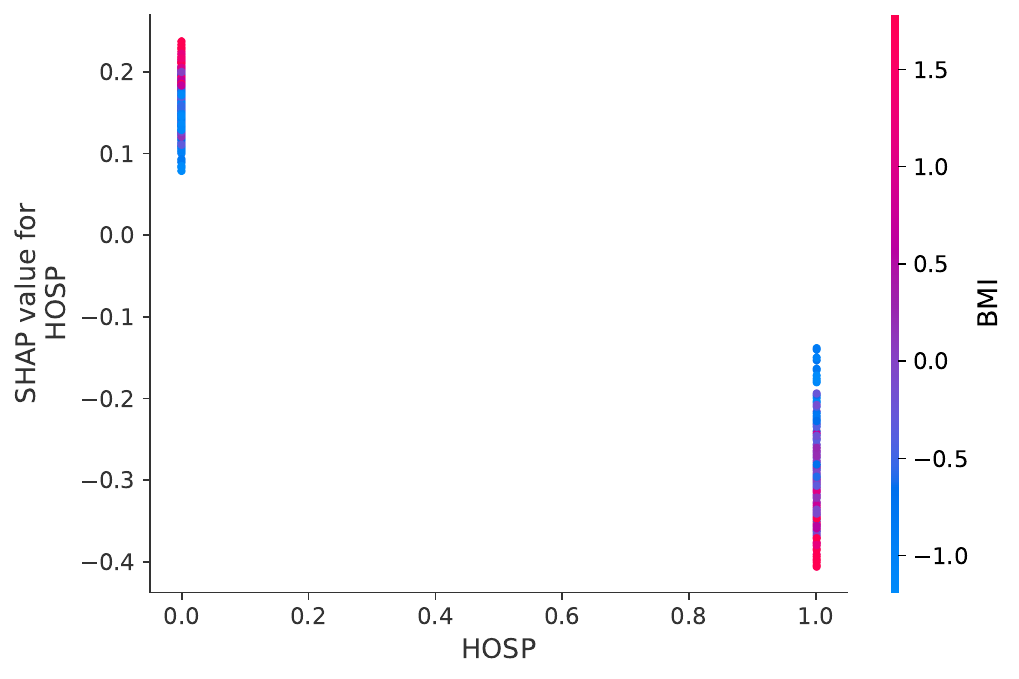}
        \caption{}
        \label{fig3:HOSP}
    \end{subfigure}
    \hfill
    \begin{subfigure}{0.48\linewidth}
        \centering
        \includegraphics[width=\linewidth]{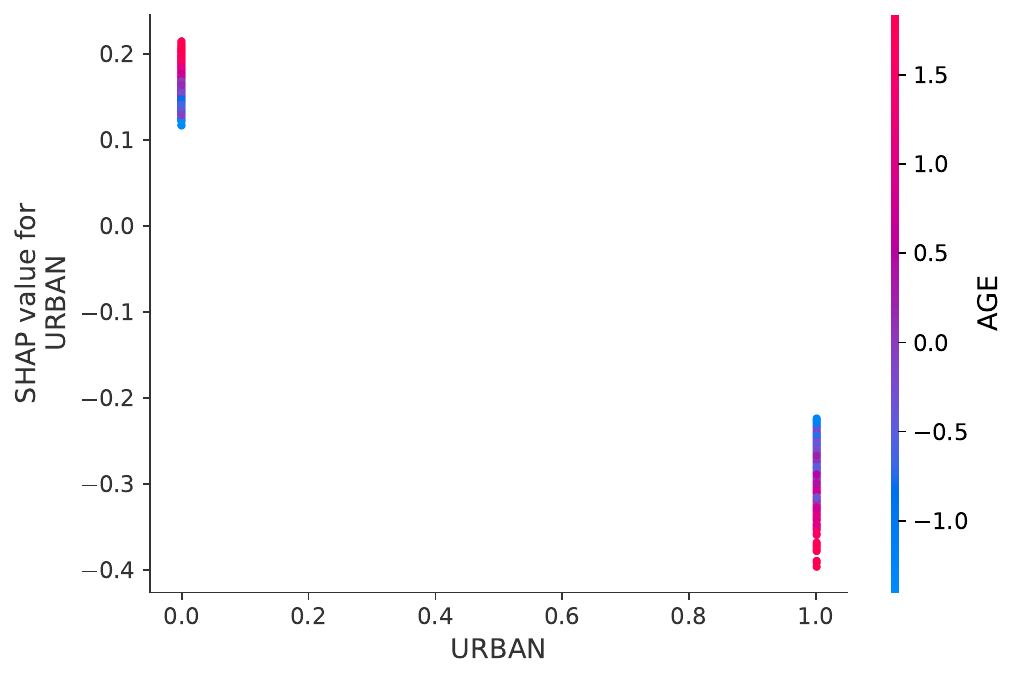}
        \caption{}
        \label{fig3:URBAN}
    \end{subfigure}

    \begin{subfigure}{0.48\linewidth}
        \centering
        \includegraphics[width=\linewidth]{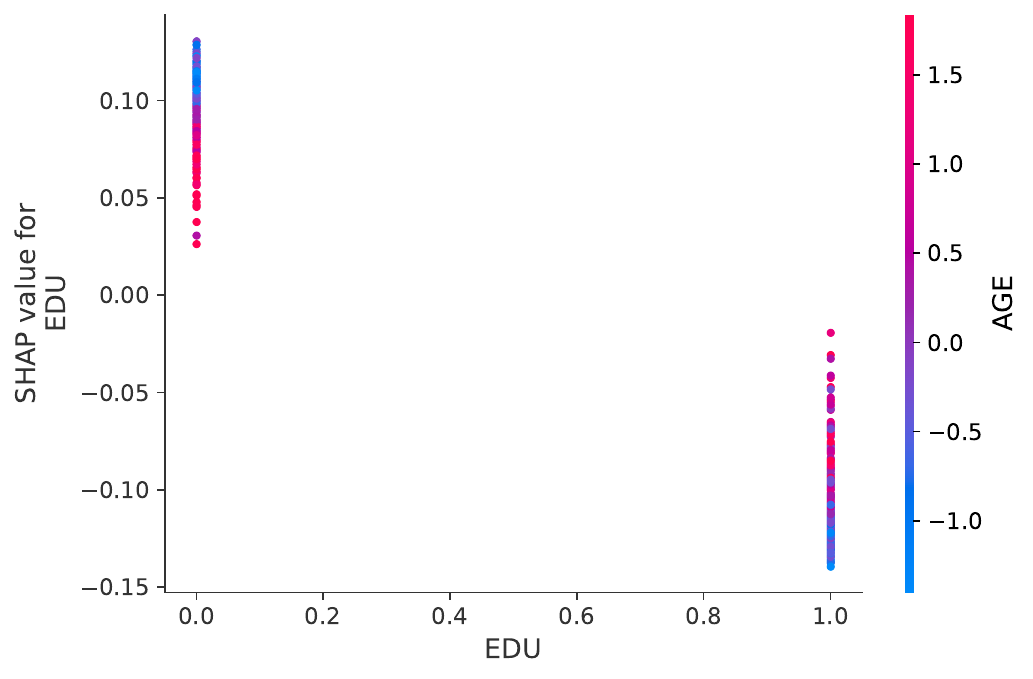}
        \caption{}
        \label{fig3:EDU}
    \end{subfigure}
    \hfill
    \begin{subfigure}{0.48\linewidth}
        \centering
        \includegraphics[width=\linewidth]{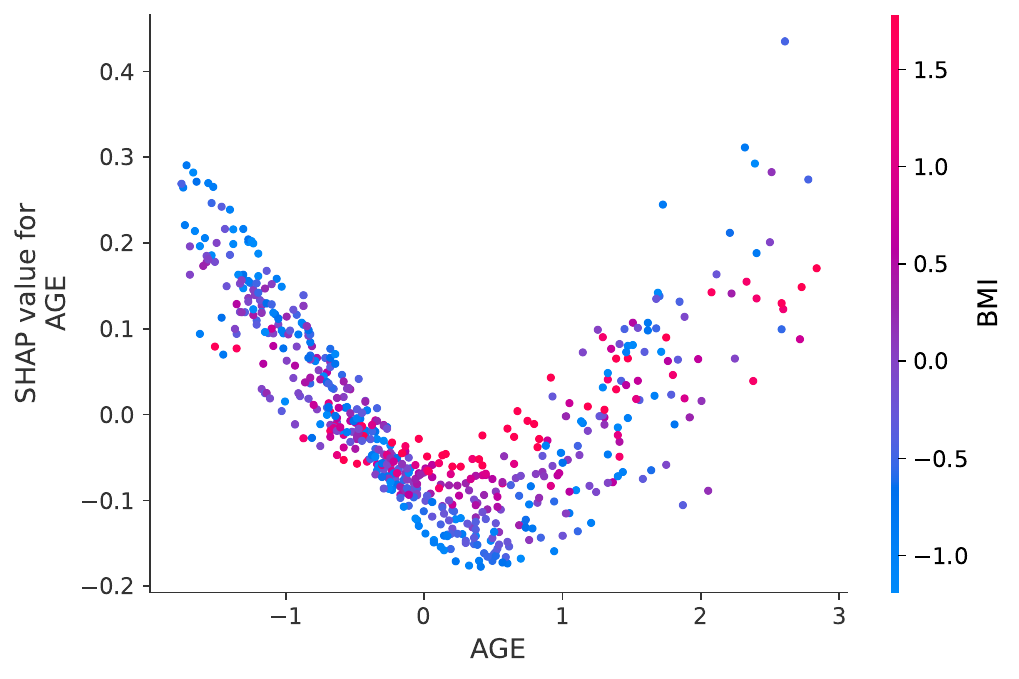}
        \caption{}
        \label{fig3:AGE}
    \end{subfigure}
    
    \caption{SHAP dependence plots for (a) HOSP, (b) URBAN, (c) EDU, and (d) AGE}
    \label{fig3:shap_dependence}
\end{figure}

We selected the top four covariates ranked by mean absolute SHAP values and plotted their SHAP dependence plots, as shown in Figure~\ref{fig3:shap_dependence}. This analysis further examines how covariates influence the model output and reveals potential interactions with other covariates. In each dependence plot, the x-axis represents the values of the covariate under consideration, while the y-axis shows its corresponding SHAP values. The color of the points indicates the values of the covariate that has the strongest interaction with the covariate under analysis, with red denoting larger values and blue denoting smaller values.

Figure~\ref{fig3:HOSP} shows that when HOSP = 1, the SHAP values for HOSP are predominantly negative, indicating that hospital delivery tends to reduce the neural network output, thereby increasing child survival probability at fixed times. Moreover, this effect is more pronounced among mothers with higher BMI (i.e., more negative SHAP values), suggesting that children of mothers with higher BMI benefit more from hospital delivery. When HOSP = 0, the SHAP values for HOSP are positive, but smaller for mothers with lower BMI, indicating that in the absence of hospital care, children of mothers with lower BMI face slightly lower risks than those of mothers with higher BMI.

Based on Figure~\ref{fig3:URBAN}, urban residents (URBAN = 1) have negative SHAP values for the URBAN covariate, indicating that living in an urban area tends to reduce the neural network output, thereby lowering the risk of child mortality. This risk-reducing effect is more pronounced among older mothers. In contrast, when URBAN = 0, the SHAP values for URBAN are generally positive, suggesting that non-urban residence increases the model prediction, implying a higher risk of child mortality. Moreover, under non-urban residence, the SHAP values for URBAN are higher for older mothers.

As shown in Figure~\ref{fig3:EDU}, when EDU = 1, EDU generally has negative SHAP values, indicating that higher maternal education tends to reduce the risk of child mortality. This risk-reducing effect is more pronounced among younger mothers (i.e., more negative SHAP values). In contrast, when EDU = 0, the covariate has positive SHAP values, suggesting that this characteristic tends to increase the model output, thereby increasing risk. Within this group, the SHAP values for EDU are higher for younger mothers.

As shown in Figure~\ref{fig3:AGE}, the SHAP values for AGE exhibit a U-shaped pattern, with negative SHAP values occurring only within a certain age range, indicating the presence of an optimal childbearing period. Maternal ages that are too young or too old increase the model output, thereby elevating the risk of child mortality. Notably, within the optimal childbearing range, the SHAP values for AGE are smaller for mothers with higher BMI, suggesting that higher BMI attenuates the risk-reducing effect associated with the optimal age period.

\section{Concluding Remarks}
\label{concl}

This paper considered a class of transformation models with unspecified covariate-dependent function and analyzed the interval-censored data arising from the (generalized) case-cohort designs. The effects of all covariates were modeled through an unknown function, which avoids overly restrictive linear assumptions and facilitates the capture of complex relationships present in real data. The present framework offers flexibility in two aspects. First, the transformation model is general, encompassing commonly used models, such as the proportional hazards model, as special cases. Second, the use of deep neural networks provides a powerful function approximation tool with strong representation learning capabilities, while mitigating the curse of dimensionality. We developed a sieve weighted likelihood estimation method that combined deep learning with Bernstein polynomials, and established the theoretical properties of the proposed estimator. The results of both simulations and real data analyses demonstrated that the proposed approach performs well in practice. Given that DNN-based survival methods often face challenges in interpretability, SHAP was employed in the real data analysis to attribute the predictions from the covariate network to the covariates, yielding several meaningful insights.

There are several potential extensions of the proposed method. First, in practice, some covariates are readily available and can cover the entire cohort. However, in this work, the model is trained only on the case-cohort sample, without fully leveraging information from the entire cohort. \cite{Zhou2024Improving} proposed an update estimation procedure that uses complete cohort information to improve estimation efficiency, but their method relies on a linear covariate assumption. Therefore, developing a deep learning approach that can incorporate entire cohort information to enhance the performance of existing methods represents a promising direction for future research. Second, most existing deep learning methods for survival analysis focus on univariate censored failure time data, and studies on bivariate data remain very limited. Modeling bivariate interval-censored data requires additional consideration of the dependence between the two failure time variables, making the extension of deep learning methods to bivariate interval-censored settings both challenging and valuable.

\bibliographystyle{apalike}
\bibliography{reference}

\section*{Appendix: Proofs of Theorems}
\label{appen}
\addcontentsline{toc}{section}{Appendix: Proofs of Theorems}

\renewcommand{\theequation}{A.\arabic{equation}}
\renewcommand{\theHequation}{A.\arabic{equation}}
\setcounter{equation}{0}

First, we introduce some notation. For any vector $\boldsymbol{v} = (v_1, \dots, v_p)^\top \in \mathbb{R}^p$, $\|\boldsymbol{v}\|_2 = (\sum_{i=1}^p v_i^2)^{1/2}$ and $\|\boldsymbol{v}\|_\infty = \max_i |v_i|$, and for any matrix $W = (w_{ij}) \in \mathbb{R}^{m \times n}$, $\|W\|_\infty = \max_{i,j} |w_{ij}|$. For any function $h$, $\|h\|_\infty$ and $\|h\|_2$ are the sup-norm and $L^2$-norm of $h$, respectively, and for any vector function $\boldsymbol{h} = (h_1, \dots, h_s)^\top$, $\|\boldsymbol{h}\|_\infty = \max_i \|h_i\|_\infty$. Denote $a_n \lesssim b_n$ as $a_n \leq c b_n$ for some $c > 0$ and any $n$. And $a_n \asymp b_n$ means $a_n \lesssim b_n$ and $b_n \lesssim a_n$. We use $C$ to denote a universal positive constant which may differ from place to place.

Let $O^{\varphi } =\left\{ K,U_1,\ldots,U_K,\Delta_{1},\ldots,\Delta_{K},\varphi  Z,\varphi \right\}$ denote a single observation, where $\varphi$ indicates whether the covariate $Z$ is observed, and $\Delta_{k}=I(U_{k-1} < T \leq  U_k)$. Define $L= \max\left\{ U_{k}: U_{k}< T, k=0, \ldots ,K\right\}$ and $R= \min\left\{ U_{k}: U_{k}\geq  T, k=1, \ldots , K+1\right\}$, with $\Delta_{L}=I(L=0)$ and $\Delta_{I}=I(L\neq 0 , R < \infty )$, where $U_0=0$ and $U_{K+1}=\infty$. Then, the observation can equivalently be represented as $O^{\varphi  } =\left\{ L,R,\Delta_{L},\Delta_{I},\varphi  Z,\varphi \right\}$. Let $\theta = (\Lambda , g)$, with the true value $\theta_0 = (\Lambda_0 , g_0)$, and define $G\left(\Lambda(t) \exp({g(Z)})\right)=G_{\theta }(t,Z)$. The weighted log-likelihood function based on a single observation $O^{\varphi}$ is then given by
\begin{equation*}
\begin{split}
l^{w}(\theta ,O^{\varphi }) &=wl(\theta ,O)\\
&=w\big\{\sum_{k=1}^{K+1}\Delta_{k} \log \left[ \exp\left(-G_{\theta }(U_{k-1},Z)\right) - \exp\left(-G_{\theta }(U_k,Z)\right) \right] \big\}\\
&= w \big\{ \Delta_{L} \log \left[ 1 - \exp\left(-G_{\theta }(R,Z)\right) \right] + \Delta_{I} \log \left[ \exp\left(-G_{\theta }(L,Z)\right) \right. \\
&\quad \left. - \exp\left(-G_{\theta }(R,Z)\right) \right] - (1 - \Delta_{L} - \Delta_{I}) G_{\theta }(L,Z) \big\},
\end{split}
\end{equation*}
where $w = \varphi / \left[(1-\Delta_{L}-\Delta_{I})p_s + (\Delta_{L}+\Delta_{I})(p_s+(1-p_s)p_c)\right]$is bounded and does not depend on $\theta$, and $p_s$ and $p_c$ are known constants. Let $O=\left\{ L,R,\Delta_{L},\Delta_{I},Z\right\}$ denote the complete data. For the purpose of the proof, let $\mathbb{P}_n$ denote the empirical measure based on $n$ independent observations, and $\mathbb{P}$ denote the true probability measure. Define $F_n(\theta )=\mathbb{P}_n l^{w}(\theta ,O^{\varphi })$ and $F(\theta )=\mathbb{P} l^{w}(\theta ,O^{\varphi })$.

Before presenting the proof, we first describe the required regularity conditions:
\begin{enumerate}
    \item[(C1)] \label{cond:C1} $H = O(\log n), s = O(n \gamma_n^2 \log n)$ and  $n \gamma_n^2 \lesssim \min(p_h)_{h=1, ..., H} \leq \max(p_h)_{h=1, ..., H} \lesssim n$.

    \item[(C2)] The nonparametric function $g_0$ is an element of $\mathcal{H}_0 = \{ g \in \mathcal{H}(L, \boldsymbol{\alpha}, \boldsymbol{d}, \boldsymbol{\tilde{d}}, B) : \mathbb{E}\{g(Z)\} = 0 \}$.

    \item[(C3)] The covariate $Z$ takes value in a bounded subset of $\mathbb{R}^{p}$ with probability density function bounded away from zero. Without loss of generality, we assume that the domain of $Z$ is taken to be $[0, 1]^p$. 

    \item[(C4)] The number of examination times $K$ is positive with $\mathbb{E}(K) < \infty$. There exists $\eta > 0$ such that $Pr(\min_{0 \leq k \leq K} (U_{k+1} - U_{k}) \geq \eta \mid K, Z) = 1$. The union of the supports of $\{U_k : k = 1, \ldots, K\}$ is contained in the interval $[c, u]$, where $0 < c < u < \infty$.

    \item[(C5)] (i) The function $\Lambda_{0}\in \mathcal{M}_0$ is continuously differentiable up to order $r$ in $[c, u]$ and satisfies $\xi^{-1} < \Lambda_{0}(c) < \Lambda_{0}(u) < \xi$ for some positive constant $\xi$. 
    (ii)The transformation $G$ is a strictly increasing function with $G(0) = 0$ and is three-times continuously differentiable in $[0, u]$.

    \item[(C6)] For every $\theta$ in a neighborhood of $\theta_0$, 
    $\mathbb{P}\{l^{w}(\theta ,O^{\varphi }) - l^{w}(\theta_0 ,O^{\varphi })\} \lesssim - d^2(\theta, \theta_0)$.
\end{enumerate}
Condition (C1) determines the structure of the neural network family $\mathcal{G}(H, s, \boldsymbol{p}, D)$ in \eqref{equa4}. Condition (C2) ensures the identifiability of the investigated model. Condition (C3)-(C6) are commonly used in the studies of interval-censored data. 

\begin{proof}[Proof of Theorem~\ref{th1}]
We first consider the estimator $\hat{\theta}^* = (\hat{\Lambda}^*, \hat{g}^*)$ in \eqref{equa6} that satisfies $\allowbreak \mathbb{E}\left\{ \hat{g}^*(Z)\right\} = \mathbb{E}\left\{ g_0(Z)\right\}$. In fact, for any estimator $\hat{\theta} = (\hat{\Lambda}, \hat{g})$, its transformation $\hat{\theta}^* = (\hat{\Lambda}\exp(\mathbb{E}\left\{ \hat{g}(Z)\right\}), \hat{g}-\mathbb{E}\left\{ \hat{g}(Z)\right\})$ is also an estimator in \eqref{equa6}.

We now prove $d(\hat{\theta}^*, \theta_0) \overset{p}{\to} 0$ as $n \to \infty$. For some $D > 0$, let $\mathcal{G}_D := \mathcal{G}(H, s, \boldsymbol{p}, D)$ and $\mathcal{M}_D = \left\{ \Lambda_n(t) = \sum_{k=0}^m \phi_k B_k(t, m, c, u) : \sum_{k=0}^m |\phi_k| \leq D , 0 \leq \phi_0 \leq \cdots \leq \phi_m \right\}$. Define
\begin{equation}
\label{equaA1}
\hat{\theta}_D^* = (\hat{\Lambda}_D^*, \hat{g}_D^*) = \operatorname*{arg\,max}_{\substack{\theta \in \mathcal{M}_D \times \mathcal{G}_D,\\ \mathbb{E}\left\{ g(Z)\right\} = \mathbb{E}\left\{ g_0(Z)\right\}}} F_n(\Lambda, g).
\end{equation}
Note that $\mathbb{P}(d(\hat{\theta}^*, \theta_0) < \infty) = 1$. Thus, it suffices to show that $d(\hat{\theta}_D^*, \theta_0) \overset{p}{\to} 0$ as $n \to \infty$ for some large enough $D$, which can be established by verifying the three conditions of Theorem 5.7 in \cite{Vaart2000}.

First, we need to verify
\begin{equation}
\label{equaA2}
\sup_{\theta \in \mathcal{M}_D \times \mathcal{G}_D} |F_n(\theta) - F(\theta)| \overset{p}{\to} 0.
\end{equation}

It suffices to show that $\mathcal{L}_n = \left\{ l^w(\theta, O^{\varphi }) = wl(\theta, O) : \theta \in \mathcal{M}_D \times \mathcal{G}_D \right\}$ is $P$-Glivenko-Cantelli. Since the exponential and logarithmic functions are Lipschitz continuous on bounded sets, and $G_{\theta}(t,Z)$ is bounded under the given regularity conditions, it follows that $l^{w}(\theta ,O^{\varphi })$ satisfies a Lipschitz condition with respect to $\theta=(\Lambda , g)$. By Theorem 2.7.11 in \cite{Vaart1996}, it remains to show that the bracketing number of $\mathcal{M}_D \times \mathcal{G}_D$ is finite. Noting that $\mathcal{M}_D$ is a class of monotone functions, invoking Lemma 6 in \cite{Zhong2022Deep}, and further applying Lemma 9.25 in \cite{Kosorok2008}, we conclude that $\mathcal{L}_n = \left\{ l^w(\theta, O^{\varphi }): \theta \in \mathcal{M}_D \times \mathcal{G}_D \right\}$ is $P$-Glivenko-Cantelli.

We now verify the second condition. According to Gibbs inequality, we obtain 
\[
\sup_{\substack{\theta: d(\theta, \theta_0) > \epsilon,\\ \mathbb{E}\left\{ g(Z)\right\} = \mathbb{E}\left\{ g_0(Z)\right\}}} F(\theta) \leq  F(\theta_0)
\]
for all $\theta \in \Theta_I=\left\{ \theta=(\Lambda, g) \in \mathcal{M}_D \times \mathcal{G}_D : \mathbb{E}\left\{ g(Z)\right\} = \mathbb{E}\left\{ g_0(Z)\right\} \right\}$. 

If 
\[
\sup_{\substack{\theta: d(\theta, \theta_0) > \epsilon,\\ \mathbb{E}\left\{ g(Z)\right\} = \mathbb{E}\left\{ g_0(Z)\right\}}} F(\theta) =  F(\theta_0)
\]
holds for some $\theta \in \Theta_I$, then there exists a sequence $\theta_m$ such that 
\[
F(\theta_m) \to \sup_{\substack{\theta: d(\theta, \theta_0) > \epsilon,\\ \mathbb{E}\left\{ g(Z)\right\} = \mathbb{E}\left\{ g_0(Z)\right\}}} F(\theta) = F(\theta_0) 
\]
and $d(\theta_m, \theta_0) > \epsilon$. Since the coefficients of the Bernstein polynomials and the parameters of the neural network are bounded, there exists a subsequence $\theta_{m'}$ of $\theta_m$, converging to $\theta_{m0}$. Because $F(\theta)$ is a continuous function of $\theta$, $F(\theta_{m0}) = F(\theta_0)$, and by the identifiability of the proposed model under the assumptions that $\Lambda$ and $G$ belong to monotone function classes and $\mathbb{E}\left\{ g(Z)\right\} = \mathbb{E}\left\{ g_0(Z)\right\}$, it follows that $\theta_{m0} = \theta_0$. However, $\theta_{m'}$ does not converge to $\theta_0$ due to the fact that $d(\theta_{m'}, \theta_0) > \epsilon$. This conflicts with the aforementioned result that $\theta_{m'}$ converges to $\theta_{m0}$. Therefore, we conclude that
\begin{equation}
\label{equaA3}
\sup_{\substack{\theta: d(\theta, \theta_0) > \epsilon,\\ \mathbb{E}\left\{ g(Z)\right\} = \mathbb{E}\left\{ g_0(Z)\right\}}} F(\theta) < F(\theta_0).
\end{equation}

Finally, we verify the third condition. Note that for any $\theta_1,\theta_2 \in \mathcal{M}_D \times \mathcal{G}_D$, it follows from the mean value theorem and some algebraic manipulations that
\begin{equation}
\label{equaA4}
\begin{split}
\mathbb{P}\left\{ l^w(\theta_1, O^{\varphi }) - l^w(\theta_2, O^{\varphi }) \right\}^{2} 
&= \mathbb{E}\left\{ l^w(\theta_1, O^{\varphi }) - l^w(\theta_2, O^{\varphi }) \right\}^{2}\\
&\lesssim \mathbb{E}\big\{ 
\Delta_{L} \left [ \log \frac{1 - \exp\left(-G_{\theta_1}(R,Z)\right)}{1 - \exp\left(-G_{\theta_2}(R,Z)\right)}\right ]^2 \\
&\quad + \Delta_{I}\left [ \log\frac{\exp\left(-G_{\theta_1}(L,Z)\right)  - \exp\left(-G_{\theta_1}(R,Z)\right)}{\exp\left(-G_{\theta_2}(L,Z)\right)  - \exp\left(-G_{\theta_2}(R,Z)\right)}\right ]^2 \\
&\quad + (1 - \Delta_{L} - \Delta_{I})\left [ G_{\theta_1}(L,Z) - G_{\theta_2}(L,Z)\right ]^2
\big\}\\
&\lesssim \mathbb{E}\left\{ \left [ G_{\theta_1}(L,Z) - G_{\theta_2}(L,Z)\right ]^{2} + \left [ G_{\theta_1}(R,Z) - G_{\theta_2}(R,Z)\right ]^{2}\right\}\\
&\lesssim \mathbb{E}\left\{ \left [ \Lambda_1(L) - \Lambda_2(L) \right ]^{2} + \left [ \Lambda_1(R) - \Lambda_2(R) \right ]^{2} + \left [ g_1(Z) - g_2(Z) \right ]^{2}\right\}\\
&= d^2(\theta_1, \theta_2).
\end{split}
\end{equation}

The Cauchy-Schwarz inequality yields
\begin{equation}
\label{equaA5}
\left |\mathbb{E}\left\{ l^w(\theta_1, O^{\varphi }) - l^w(\theta_2, O^{\varphi })\right\} \right |\leq \left [ \mathbb{E}\left\{ l^w(\theta_1, O^{\varphi }) - l^w(\theta_2, O^{\varphi }) \right\}^{2} \right ]^{\frac{1}{2}}\lesssim d(\theta_1, \theta_2).
\end{equation}

Define $g_{n0} = \underset{g \in \mathcal{G}(H, s, \boldsymbol{p},D/2)}{\arg \min} \|g - g_{0}\|_{L^2}$. By the proof of Theorem 1 in \cite{Schmidt-Hieber2020deep}, we have $\|g_{n0} - g_0\|_{L^2} = O(\gamma_n \log^2 n)$. Let $g_{n0}^* = g_{n0} - \mathbb{E}\{g_{n0}(Z)\}$. Clearly, $g_{n0}^* \in \mathcal{G}_D$ and $\|g_{n0}^* - g_0\|_{L^2} = \|g_{n0} - g_0 - \mathbb{E}\{g_{n0}(Z) - g_0(Z)\}\|_{L^2} \lesssim \|g_{n0} - g_0\|_{L^2} = O(\gamma_n \log^2 n).$ Furthermore, Theorem 1.6.2 of \cite{Lorentz1986} states that there exists a Bernstein polynomial $\Lambda_{n0}$ such that $\|\Lambda_{n0} - \Lambda_0\|_\infty = O(m^{-r/2}) = O(n^{-r\nu/2})$, which in turn implies $\|\Lambda_{n0} - \Lambda_0\|_2 = O(n^{-r\nu /2})$. Then, by \eqref{equaA2}, \eqref{equaA5} and the law of large numbers, we have
\begin{equation*}
\begin{split}
|F_n(\Lambda_{n0}, g_{n0}^*) - F_n(\Lambda_0, g_0)| \leq& |F_n(\Lambda_{n0}, g_{n0}^*) - F(\Lambda_{n0}, g_{n0}^*)| + |F(\Lambda_{n0}, g_{n0}^*) - F(\Lambda_0, g_{n0}^*)| \\
& + |F(\Lambda_0, g_{n0}^*) - F(\Lambda_0, g_0)| + |F(\Lambda_0, g_0) - F_n(\Lambda_0, g_0)|\\ 
=& o_p(1).
\end{split}
\end{equation*}

Since $\hat{\theta}_D^*$ is the maximizer of \eqref{equaA1}, we obtain $F_n(\hat{\Lambda}_D^*, \hat{g}_D^*) \geq F_n(\Lambda_{n0}, g_{n0}^*) = F_n(\Lambda_0, g_0) - o_p(1)$, which gives
\begin{equation}
\label{equaA6}
F_n(\hat{\theta }_D^*) \geq F_n(\theta _0) - o_p(1).
\end{equation}
Therefore, the conditions of Theorem 5.7 in \cite{Vaart2000} follow from \eqref{equaA2}, \eqref{equaA3}, and \eqref{equaA6}, which implies that $d(\hat{\theta}_D^*, \theta_0) \overset{p}{\to} 0$ as $n \to \infty$.
\end{proof}

\begin{proof}[Proof of Theorem~\ref{th2}]
We prove this theorem by applying Theorem 3.4.1 of \cite{Vaart1996}. Define $\theta_{n0}=(\Lambda_{n0}, g_{n0}^*)$. From the proof of Theorem 1 it follows that $d(\theta_0, \theta_{n0}) = O(n^{-r\nu/2} + \gamma_n \log^2 n)$. For any $\delta > 0$, let $\mathcal{A}_\delta = \{ \theta  = (\Lambda , g) \in \mathcal{M}_D \times \mathcal{G}_D : \delta /2 <  d(\theta, \theta_{n0}) \leq \delta \}$. One can easily show that 
\begin{equation*}
F(\theta_0) - F(\theta_{n0}) = \mathbb{P}\{l^w(\theta_0, O^\varphi) - l^w(\theta_{n0}, O^\varphi)\} \lesssim d(\theta_0, \theta_{n0}) \lesssim n^{-r\nu/2} + \gamma_n \log^2 n.
\end{equation*}
By Condition (C6), we have for large $n$, 
\begin{equation*}
F(\theta) - F(\theta_{n0}) = F(\theta) - F(\theta_{0}) + F(\theta_{0}) - F(\theta_{n0}) \leq -C\delta ^2 + C(n^{-r\nu/2} + \gamma_n \log^2 n) = -C\delta^2,
\end{equation*}
for any $\theta \in \mathcal{A}_\delta$, which implies
\begin{equation*}
 \sup_{\theta \in \mathcal{A}_\delta} [F(\theta) - F(\theta_{n0})] \lesssim -\delta^2.
\end{equation*}

By Lemma~\ref{lem1}, we know that
\begin{equation*}
\mathbb{E}^* \sup_{\theta \in \mathcal{A}_\delta} \sqrt{n} |(F_n - F)(\theta) - (F_n - F)(\theta_{n0})| \lesssim \phi_n(\delta),
\end{equation*}
where $\phi_n(\delta) = \delta\sqrt{s\log\frac{Q}{\delta}} + \frac{s}{\sqrt{n}}\log\frac{Q}{\delta}$ with $Q = H \prod_{h=0}^{H} (p_h + 1) \sum_{h=0}^{H} p_h p_{h+1}$. It is easy to see that $\phi_n(\delta)/\delta$ is decreasing in $\delta$. 

Denote $r_n = \gamma_n \log^2 n$. By Condition (C1), it is clear that  
\begin{equation*}
r_n^{-2} \phi_n(r_n) \leq \sqrt{n}.
\end{equation*}

Finally, note that $F_n(\hat{\theta}_D^*) - F_n(\theta_{n0})\geq 0$ and $d(\hat{\theta}_D^*, \theta_{n0}) \leq d(\hat{\theta}_D^*, \theta_0) + d(\theta_0, \theta_{n0}) \to 0$ in probability. Hence, by applying Theorem 3.4.1 of \cite{Vaart1996}, we have 
\begin{equation*}
d(\hat{\theta}_D^*, \theta_{n0}) = O_p(r_n).
\end{equation*}
This gives $d(\hat{\theta}^*, \theta_{n0}) = O_p(r_n)$, together with $d(\theta_{n0}, \theta_0) = O(n^{-r\nu/2} + \gamma_n \log^2 n)$, yields that
\begin{equation*}
d(\hat{\theta}^*, \theta_{0}) = O_p(n^{-r\nu/2} + \gamma_n \log^2 n).
\end{equation*}

Furthermore, we know $\|\hat{g}^* - g_{n0}^*\|_{L^2} = O_p(r_n)$, together with $\|g_{n0}^* - g_0\|_{L^2} = O(\gamma_n \log^2 n)$, yields that
\begin{equation*}
\|\hat{g}^* - g_{0}\|_{L^2} = O_p(\gamma_n \log^2 n).
\end{equation*}
\end{proof}

\begin{proof}[Proof of Theorem~\ref{th3}]
Let $P_{(\Lambda_0, g_0)}$ be the probability distribution determined by the cumulative baseline hazard function $\Lambda_0$ and nonparametric function $g_0$. Denote $\mathcal{P}_0 = \{P_{(\Lambda_0, g_0)} : \Lambda_0 \in \mathcal{M}_0, g_0 \in \mathcal{H}_0\}$ and $\mathcal{P}_1 = \{P_{(\Lambda_0, g_0)} : \Lambda_0 \in \mathcal{M}_1, g_0 \in \mathcal{H}_1\}$, where $\mathcal{M}_1 = \{\Lambda \in \mathcal{M}_0 : \Lambda(u) - \Lambda(c) = 1\}$ and $\mathcal{H}_1 = \mathcal{H}(L, \boldsymbol{\alpha}, \boldsymbol{d}, \tilde{\boldsymbol{d}}, B/2)$. For any $(\Lambda_1, g_1) \in \mathcal{M}_1 \times \mathcal{H}_1$, it holds that $P_{(\Lambda_1, g_1)} \overset{d}{=} P_{(\Lambda_1 \exp(\mu), g_1 - \mu)}$ and $P_{(\Lambda_1 \exp(\mu), g_1 - \mu)} \in \mathcal{P}_0$, where $\mu = \mathbb{E}\{g_1(Z)\}$ and $P_1 \overset{d}{=} P_2$ means $P_1$ and $P_2$ have the same probability measure. In other words, $\mathcal{P}_1$ can be viewed as a subset of $\mathcal{P}_0$. Moreover, if $\hat{g}_1$ is an estimator of $g_1 \in \mathcal{H}_1$ based on the observed data $\left\{K_i,U_{i1},...,U_{iK_i},\Delta_{i1},...,\Delta_{iK_i},\varphi_i  Z_i,\varphi_i \right\}, i = 1, \ldots , n$ under some model $P_{(\Lambda_1, g_1)} \in \mathcal{P}_1$, then $\hat{g}_0 := \hat{g}_1 - \mu$ with $\mu = \mathbb{E}\{g_1(Z)\}$ is also an estimator of $g_0 := g_1 - \mu$ based on same copies of the observed data under $P_{(\Lambda_1 \exp(\mu), g_0)} (\overset{d}{=}P_{(\Lambda_1, g_1)}) \in \mathcal{P}_0$. It follows directly that $\hat{g}_1 - g_1 = \hat{g}_0 - g_0$, and therefore
\begin{equation}
\label{equaA7}
\begin{split}
\inf_{\hat{g}_0} & \sup_{(\Lambda_0, g_0) \in \mathcal{M}_0 \times \mathcal{H}_0} \mathbb{E}_{P_{(\Lambda_0, g_0)}} \left\{ \hat{g}_0(Z) - g_0(Z) \right\}^2 \\
&\geq \inf_{\hat{g}_1} \sup_{(\Lambda_1, g_1) \in \mathcal{M}_1 \times \mathcal{H}_1} \mathbb{E}_{P_{(\Lambda_1, g_1)}} \left\{ \hat{g}_1(Z) - g_1(Z) \right\}^2,
\end{split}
\end{equation}
where $\mathbb{E}_P$ is the expectation under the distribution $P$ and the infimum is taken over all possible estimators $\hat{g}_0$ and $\hat{g}_1$ based on the observed data under the probabilities in $\mathcal{P}_0$ and $\mathcal{P}_1$, respectively.

Consequently, it suffices to derive a lower bound for the right-hand side of \eqref{equaA7}, which simultaneously provides a lower bound for the left-hand side of \eqref{equaA7}.

For $\Lambda_0 \in \mathcal{M}_1$ and $ g^{(0)},  g^{(1)} \in \mathcal{H}_1$, let $\tilde{P}_0$ and $\tilde{P}_1$ be the joint probability distribution of the observed data $\left\{K_i,U_{i1},...,U_{iK_i},\Delta_{i1},...,\Delta_{iK_i},\varphi_i  Z_i,\varphi_i \right\}, i = 1, \ldots , n$ under $P_{(\Lambda_0, g^{(0)})}$ and $P_{(\Lambda_0, g^{(1)})}$, respectively. Correspondingly, let $P_0$ and $P_1$ denote the joint probability distributions of the complete data under $P_{(\Lambda_0, g^{(0)})}$ and $P_{(\Lambda_0, g^{(1)})}$, respectively.

The Kullback-Leibler distance between $P_1$ and $P_0$ is
\begin{align*}
\text{KL}(P_1, P_0) &= \mathbb{E}_{P_1} \log \frac{P_1}{P_0} = \mathbb{E}_{\tilde{P}_1} w\log \frac{P_1}{P_0} \\
&= n \mathbb{E}_{\tilde{P}_1} w \big\{ [1 - \exp(-G_{\theta^{(1)}}(U_1,Z))] \log \left[ \frac{1 - \exp(-G_{\theta^{(1)}}(U_1,Z))}{1 - \exp(-G_{\theta^{(0)}}(U_1,Z))} \right] \\
&\quad + \sum_{k=2}^{K}[\exp(-G_{\theta^{(1)}}(U_{k-1},Z)) - \exp(-G_{\theta^{(1)}}(U_k,Z))] \\
&\quad \times \log \left[ \frac{\exp(-G_{\theta^{(1)}}(U_{k-1},Z)) - \exp(-G_{\theta^{(1)}}(U_k,Z))}{\exp(-G_{\theta^{(0)}}(U_{k-1},Z)) - \exp(-G_{\theta^{(0)}}(U_k,Z))} \right] \\
&\quad + \exp(-G_{\theta^{(1)}}(U_K,Z)) \log \left[ \frac{\exp(-G_{\theta^{(1)}}(U_K,Z))}{\exp(-G_{\theta^{(0)}}(U_K,Z))} \right] \big\} \\
&= n \mathbb{E}_{\tilde{P}_1} w \big\{ [1 - \exp(-G_{\theta^{(0)}}(U_1,Z))] h \left[ \frac{1 - \exp(-G_{\theta^{(1)}}(U_1,Z))}{1 - \exp(-G_{\theta^{(0)}}(U_1,Z))} \right] \\
&\quad + \sum_{k=2}^{K}[\exp(-G_{\theta^{(0)}}(U_{k-1},Z)) - \exp(-G_{\theta^{(0)}}(U_k,Z))] \\
&\quad \times h \left[ \frac{\exp(-G_{\theta^{(1)}}(U_{k-1},Z)) - \exp(-G_{\theta^{(1)}}(U_k,Z))}{\exp(-G_{\theta^{(0)}}(U_{k-1},Z)) - \exp(-G_{\theta^{(0)}}(U_k,Z))} \right] \\
&\quad + \exp(-G_{\theta^{(0)}}(U_K,Z)) h \left[ \frac{\exp(-G_{\theta^{(1)}}(U_K,Z))}{\exp(-G_{\theta^{(0)}}(U_K,Z))} \right] \big\},
\end{align*}
where $\theta^{(0)} = (\Lambda_0 , g^{(0)})$, $\theta^{(1)} = (\Lambda_0, g^{(1)})$, $G\left(\Lambda(U_k) \exp(g(Z))\right)=G_{\theta }(U_k,Z)$ and $h(x) = x \log x - x + 1$. Since $h(x) \leq (x - 1)^2$ for all $x>0$, it follows from the mean value theorem and some algebraic manipulations that
\begin{equation}
\label{equaA8}
\text{KL}(P_1, P_0) \leq Cn \|g^{(1)} - g^{(0)}\|^2_{L^2}.
\end{equation}

By the proof of Theorem 3 in \cite{Schmidt-Hieber2020deep}, there exist $g^{(0)}, \ldots, g^{(M)} \in \mathcal{H}_1$ and constant $C_1, C_2 > 0$, such that
\begin{equation}
\label{equaA9}
\|g^{(j)} - g^{(k)}\|_{L^2} \geq 2C_1 \gamma_n > 0
\end{equation}
and
\begin{equation}
\label{equaA10}
\frac{Cn}{M} \sum_{j=1}^{M} \|g^{(j)} - g^{(0)}\|^2_{L^2} \leq C_2 \log M.
\end{equation}

Combining \eqref{equaA8}, \eqref{equaA9}, and \eqref{equaA10}, it follows from Theorem 2.5 in \cite{Tsybakov2009} that
\begin{equation*}
\inf_{\hat{g}_1} \sup_{g_1 \in \mathcal{H}_1} \mathbb{P} \left( \|\hat{g}_1 - g_1\|_{L^2} \geq C_1 \gamma_n \right) \geq \frac{\sqrt{M}}{1 + \sqrt{M}} \left( 1 - 2C_2 - \sqrt{\frac{2C_2}{\log M}} \right).
\end{equation*}
This establishes that
\begin{equation*}
\inf_{\hat{g}_1} \sup_{(\Lambda_1, g_1) \in \mathcal{M}_1 \times \mathcal{H}_1} \mathbb{E}_{P_{(\Lambda_1, g_1)}} \left\{ \hat{g}_1(Z) - g_1(Z) \right\}^2 \geq C_3 \gamma_n^2,
\end{equation*}
for some constant $0 < C_3 < \infty$. Therefore, the proof is completed.
\end{proof}

The next lemma serves as an auxiliary result in the proof of Theorem~\ref{th2}.
\begin{lemma}
\label{lem1}
Let $\mathcal{B}_\delta = \{\theta  = (\Lambda , g) \in \mathcal{M}_D \times \mathcal{G}_D : \|\Lambda - \Lambda_{n0}\|_{2} \leq \delta, \|g - g_{n0}^*\|_{L^2} \leq \delta\}$. Define $\mathbb{G}_n = \sqrt{n}(\mathbb{P}_n - \mathbb{P})$, then
\begin{equation*}
\mathbb{E}^* \sup_{\theta \in \mathcal{B}_\delta} \left| \mathbb{G}_n\{l^{w}(\theta ,O^{\varphi }) - l^{w}(\theta_{n0} ,O^{\varphi })\} \right| = O \left( \delta \sqrt{s \log \frac{Q}{\delta}} + \frac{s}{\sqrt{n}} \log \frac{Q}{\delta} \right),
\end{equation*}
where $\mathbb{E}^*$ is the outer measure and $Q = H \prod_{h=0}^{H} (p_h + 1) \sum_{h=0}^{H} p_h p_{h+1}$.
\end{lemma}
\begin{proof}[Proof]
Let $\mathcal{F}_\mathcal{B}(\delta) = \{l^{w}(\theta ,O^{\varphi }) - l^{w}(\theta_{n0} ,O^{\varphi }): \theta \in \mathcal{B}_\delta\}$ and  $\|\mathbb{G}_n\|_{\mathcal{F}_\mathcal{B}(\delta)} = \sup_{f \in \mathcal{F}_\mathcal{B}(\delta)}|\mathbb{G}_n f| = \sup_{\theta \in \mathcal{B}_\delta}|\mathbb{G}_n\{l^{w}(\theta ,O^{\varphi }) - l^{w}(\theta_{n0} ,O^{\varphi })\}|$. Note that, for any $\theta_1, \theta_2 \in \mathcal{B}_\delta$, by an argument similar to that in \eqref{equaA4}, it can be shown that
\begin{equation*}
\mathbb{E}\left\{ l^w(\theta_1, O^{\varphi }) - l^w(\theta_2, O^{\varphi }) \right\}^{2} \lesssim d^2(\theta_1, \theta_2).
\end{equation*}

Define $\mathcal{M}_\delta  = \{\Lambda \in \mathcal{M}_D : \|\Lambda - \Lambda_{n0}\|_{2} \leq \delta\}$. According to \cite{Shen1994Convergence} on page 597 and Lemma 9.22 in \cite{Kosorok2008}, for any $\delta > 0$ and $0 < \epsilon < \delta$, we have $\log N_{[\,]}(\epsilon, \mathcal{M}_{\delta}, L_2(\mathbb{P})) \lesssim (m + 1) \log \frac{\delta }{\epsilon}$. Then, by Lemma 6 in \cite{Zhong2022Deep}, it follows, if $m+1 \leq s$ and $\delta \leq Q$,
\begin{equation*}
\log N_{[\,]}(\epsilon, \mathcal{F}_\mathcal{B}(\delta), L_2(\mathbb{P})) \lesssim (m+1) \log \frac{\delta}{\epsilon} + s \log \frac{Q}{\epsilon} \lesssim s \log \frac{Q}{\epsilon}.
\end{equation*}
Moreover, we obtain
\begin{equation*}
\begin{split}
J_{[\,]}(\delta, \mathcal{F}_\mathcal{B}(\delta), L_2(\mathbb{P})) &:= \int_{0}^{\delta } \sqrt{1 + \log N_{[\,]}(\epsilon, \mathcal{F}_\mathcal{B}(\delta), L_2(\mathbb{P}))}  d\epsilon \\
&\lesssim \int_{0}^{\delta} \sqrt{1 +  s \log \frac{Q}{\epsilon}}  d\epsilon \\
&= \sqrt{\frac{s}{2}} Q \int_{\sqrt{2 \log \frac{Q}{\delta}}}^{\infty} v^2 e^{-v^2/2}  dv \\
&\asymp \delta \sqrt{s \log \frac{Q}{\delta}}.
\end{split}
\end{equation*}
By Lemma 3.4.2 of \cite{Vaart1996}, it follows that
\begin{equation*}
\begin{split}
\mathbb{E}^{*}\|\mathbb{G}_n\|_{\mathcal{F}_\mathcal{B}(\delta)} &\lesssim J_{[\,]}(\delta, \mathcal{F}_\mathcal{B}(\delta), L_2(\mathbb{P}))\left\{1 + \frac{J_{[\,]}(\delta, \mathcal{F}_\mathcal{B}(\delta), L_2(\mathbb{P}))}{\delta^{2}\sqrt{n}}\right\} \\
&\lesssim \delta\sqrt{s\log\frac{Q}{\delta}} + \frac{s}{\sqrt{n}}\log\frac{Q}{\delta}.
\end{split}
\end{equation*}
\end{proof}

\end{document}